\RequirePackage{tikz}
\usetikzlibrary{calc, positioning, snakes,calligraphy}

\documentclass{sn-jnl}
\usepackage[authoryear]{natbib}
\renewcommand{\cite}[1]{\citep{#1}}

\normalbaroutside

\usepackage[utf8]{inputenc}
\usepackage[english]{babel}
\usepackage{graphicx,float,subcaption,placeins}
\usepackage{amsmath,amssymb,amsthm,bbm}
\usepackage{bm}
\usepackage{braket}
\usepackage{mathtools}
\usepackage{subcaption}
\usepackage{xspace}
\usepackage[mode=buildnew]{standalone}
\usepackage{soulpos}
\usepackage[capitalize]{cleveref}
\usepackage{scalerel}[2016/12/29]

\makeatletter
\newcommand{\oset}[3][0ex]{%
	\mathrel{\mathop{#3}\limits^{
			\vbox to#1{\kern-2\ex@
				\hbox{$\scriptstyle#2$}\vss}}}}
\makeatother

\definecolor{vir1of4}{RGB}{68,1,84}
\definecolor{vir2of4}{RGB}{49,104,142}
\definecolor{vir3of4}{RGB}{53,183,121}
\definecolor{vir4of4}{RGB}{253,231,37}

\newcommand{\eg}{e.\,g.\@\xspace}
\newcommand{\ie}{i.\,e.\@\xspace}
\newcommand{\bE}{\boldsymbol{E}}
\newcommand{\bQ}{\boldsymbol{Q}}
\newcommand{\bx}{\boldsymbol{x}}

\renewcommand{\vec}[1]{\bm{#1}}
\newcommand{\mat}[1]{\bm{#1}}
\newcommand{\upperq}[2]{\hat{q}_{#1}^{#2}}
\newcommand{\lowerq}[2]{\check{q}_{#1}^{#2}}
\newcommand{\sortQ}[1]{q_{#1}}
\newcommand{\maxD}[1]{\hat D(#1)}
\newcommand{\minD}[1]{\check D(#1)}
\newcommand{\singleton}[1]{\lbrace #1 \rbrace}
\newcommand{\given}{\,|\,}
\newcommand{\wuo}{y}
\newcommand{\wud}{d}
\newcommand{\wuh}{w}
\newcommand{\drs}[1]{D^*_{#1}}
\newcommand{\wut}{w}
\DeclarePairedDelimiter\abs{\lvert}{\rvert}%
\DeclarePairedDelimiter\norm{\lVert}{\rVert}%
\DeclarePairedDelimiter{\nint}{\lfloor}{\rceil}

\DeclareMathOperator*{\argmin}{arg\,min}
\DeclareMathOperator{\sign}{sign}

\newcommand{\BB}{\mathbb{B}}
\newcommand{\QUBOSET}{\mathcal{Q}}
\newcommand{\DR}{\mathsf{DR}}
\newcommand{\GRE}{\mathsf{G}}
\newcommand{\MOR}{\mathsf{M}}

\newcommand{\QUBO}{\textsc{Qubo}\@\xspace}
\newcommand{\BC}{\textsc{BinClustering}\@\xspace}
\newcommand{\SuS}{\textsc{SubsetSum}\@\xspace}

\newcommand*{\defeq}{\mathrel{\vcenter{\baselineskip0.5ex\lineskiplimit0pt\hbox{\scriptsize.}\hbox{\scriptsize.}}}%
	=}
\newcommand*{\eqdef}{=\mathrel{\vcenter{\baselineskip0.5ex\lineskiplimit0pt\hbox{\scriptsize.}\hbox{\scriptsize.}}}}
\newcommand{\optleq}{\oset[-1.34ex]{\scaleobj{0.75}{\ast}}{\sqsubseteq}}
\newcommand{\optsetsymb}{S^*}
\newcommand{\optset}[1]{\optsetsymb({#1})}
\newcommand{\suboptset}[2]{S_{#1}^*(#2)}
\newcommand{\diffset}[1]{D( #1 )}
\newcommand{\elemset}[1]{\mathcal{U}( #1 )}

\newcommand{\interv}[5]{%
	\draw ($(#1,#2)+(0,-2pt)$) -- ++(0,4pt) node[anchor=south,yshift=-2pt] {\tiny #4};%
	\draw (#1,#2) -- ++(#3,0);%
	\draw ($(#1,#2)+(#3,-2pt)$) -- ++(0,4pt) node[anchor=south,yshift=-2pt] {\tiny #5};}

\newcommand{\intervb}[5]{%
	\draw[color=vir3of4] ($(#1,#2)+(0,-3pt)$) -- ++(0,6pt) node[anchor=south,yshift=-2pt] {\tiny #4};%
	\draw[line width=4pt,color=vir3of4] (#1,#2) -- ++(#3,0);%
	\draw[color=vir3of4] ($(#1,#2)+(#3,-3pt)$) -- ++(0,6pt) node[anchor=south,yshift=-2pt] {\tiny #5};}

\newcommand{\intervbc}[5]{%
	\draw[color=#4] ($(#1,#2)+(0,-3pt)$) -- ++(0,6pt);%
	\draw[line width=4pt,color=#4] (#1,#2) -- ++(#3,0)
	node[anchor=west,color=black] {\tiny #5};%
	\draw[color=#4] ($(#1,#2)+(#3,-3pt)$) -- ++(0,6pt);}

\theoremstyle{thmstyleone}
\newtheorem{theorem}{Theorem}
\newtheorem{lemma}{Lemma}

\newtheorem{proposition}{Proposition}

\Crefname{definition}{Def.}{Defs.}
\Crefname{proposition}{Prop.}{Props.}

\theoremstyle{thmstyletwo}%
\newtheorem{example}{Example}%

\theoremstyle{thmstylethree}%
\newtheorem{definition}{Definition}%

\begin{document}
	
	\title[QUBO Compression]{Optimum-Preserving QUBO Parameter Compression}
	
	\author*[1]{\fnm{Sascha} \sur{M\"{u}cke}}\email{sascha.muecke@tu-dortmund.de}
	
	\author[2]{\fnm{Thore} \sur{Gerlach}}
	
	\author[2]{\fnm{Nico} \sur{Piatkowski}}
	
	\affil*[1]{\orgname{TU Dortmund}, \orgdiv{AI Group}, \orgaddress{\city{Dortmund}, \postcode{44227}, \country{Germany}}}
	
	\affil[2]{\orgname{Fraunhofer IAIS}, \orgaddress{\city{Sankt Augustin}, \postcode{53757}, \country{Germany}}}
	
	\abstract{
		Quadratic unconstrained binary optimization (\QUBO) problems are well-studied, not least because they can be approached using contemporary quantum annealing or classical hardware acceleration.
		However, due to limited precision and hardware noise, the effective set of feasible parameter values is severely restricted. As a result, otherwise solvable problems become harder or even intractable. 
		In this work, we study the implications of solving \QUBO problems under limited precision. 
		Specifically, it is shown that the problem's dynamic range has a crucial impact on the problem's robustness against distortions. 
		We show this by formalizing the notion of preserving optima between \QUBO instances and explore to which extend parameters can be modified without changing the set of minimizing solutions.
		Based on these insights, we introduce techniques to reduce the dynamic range of a given \QUBO instance based on theoretical bounds of the minimal energy value.
		An experimental evaluation on random \QUBO instances as well as \QUBO-encoded \BC and \SuS problems show that our theoretical findings manifest in practice. 
		Results on quantum annealing hardware show that the performance can be improved drastically when following our methodology.
	}
	
	
	\maketitle
	
	\section{Introduction}
	Over the past decades, intensive research has been conducted on the exploitation of quantum effects for problem solving and faster computing \cite{grover1996,shor1997,kadowaki1998,farhi2000,vandam2001,farhi2014,peruzzo2014}.
	In the past years in particular, quantum machine learning (QML) emerged as a crossover field between machine learning (ML) and quantum computing (QC) \cite{lloyd2013,schuld2015,dunjko2016,biamonte2017}.
	One of several QC paradigms is adiabatic quantum computing (AQC), which exploits quantum tunneling effects to approximate the ground state of a parametric Hamiltonian, commonly of an Ising model \cite{kadowaki1998,farhi2000}: 
	In contemporary hardware implementations of AQC, the Hamiltonian can be written as quadratic function over binary vectors $\bx\in\lbrace 0,1\rbrace^n$, for which it tries to find the minimum energy eigenstate $\ket{\bx^*}$. 
	This procedure is equivalent to solving the \emph{quadratic unconstrained binary optimization} (\QUBO) problem, which has been investigated since the 1960s (see \eg \cite{kochenberger2014}). 
	Its value lies in its applicability to a wide range of combinatorial optimization problems, from economics \cite{laughhunn1970,hammer1971} over satisfiability \cite{kochenberger2005}, resource allocation and routing problems \cite{neukart2017,stollenwerk2019} to machine learning \cite{bauckhage2018,mucke2019a,date2020,mucke2023}---just to name a few.
	Solving \QUBO is, in general, \textsf{NP}-hard \cite{pardalos1992}.
	
	Despite the promise of quantum speedup, currently available devices that claim to perform AQC could not be proven to be faster than classical computing resources yet \cite{ronnow2014}. 
	Classical hardware-based \QUBO solvers have been developed as alternatives to imperfect quantum devices \cite{matsubara2017,mucke2019b} to facilitate both, research and practical applications. 
	A problem that hardware solvers have in common is the limited parameter precision: 
	While in theory \QUBO is defined with parameters in $\mathbb{R}$, digital devices rely on finite number representations that necessarily sacrifice some precision, as $B$ bits can only encode $2^B$ distinct values. 
	The well-known solution to this problem in the classical realm is floating-point arithmetic, e.g., IEEE 754 \cite{iso2020}. 
	
	A similar problem specific to D-Wave's quantum annealers is ``integrated control errors'' (ICE) \cite{d-wavesystems2021, dwaveice2023}, which randomly distort the Hamiltonian, leading to a skewed energy landscape.
	Depending on the particular problem instance, this distortion may change the optimum (see \cref{fig:perturbation}). Classical solutions like IEEE 754 are not applicable, due to the physical, analogue representation of the Hamiltonian. 
	
	\begin{figure}
		\centering
		\includegraphics[width=\textwidth]{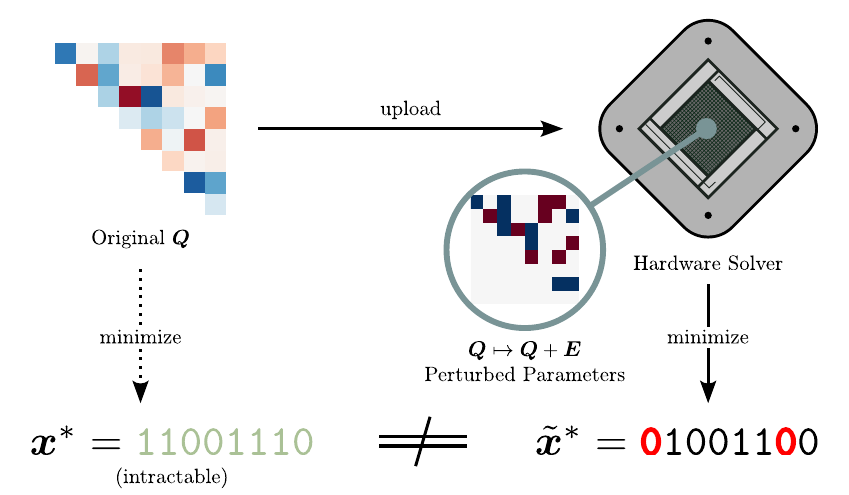}
		\caption{\QUBO solvers have a limited parameter resolution, leading to perturbations $\mat E$ that may result in false optima.}
		\label{fig:perturbation}
	\end{figure}
	
	Finally, similar effects occur also in digital annealing hardware, due to rounding of problem parameters, \eg, when converting between floating point formats. 
	To the best of our knowledge, no research has been conducted on the effect that rounding or otherwise distorting \QUBO or Ising model parameters has on the energy landscape and the distribution of optima.
	
	\begin{figure}[t]
		\centering
		\includegraphics[width=\textwidth]{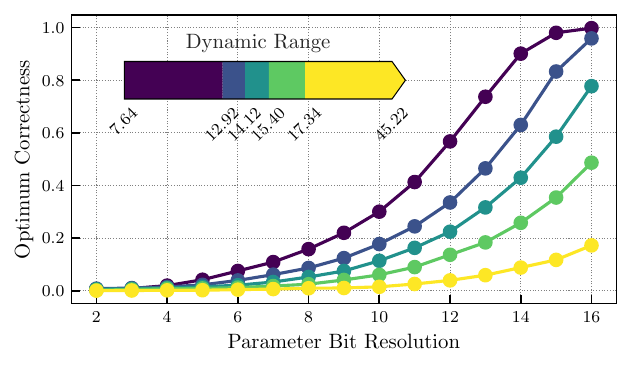}
		\caption{The effect of rounding \QUBO parameters to a certain number of bits on the probability that the optimum remains intact.
			The \QUBO matrices represent random instances of the \SuS problem with $n=16$.
			Each line represents the mean taken over 20,000 instances from a specific bin containing problems with a similar dynamic range (DR, see \cref{sec:dr}).
			For a detailed description of the \SuS problem and our methodology see \cref{sec:rounding-error-method}.
			We observe that the probability of retaining the correct optimum increases (1) with an increasing number of bits, and (2) with decreasing DR.}
		\label{fig:rounding-error}
	\end{figure}
	
	However, we find that rounding can drastically deteriorate the probability to find the global optimum, depending on the parameter distribution:
	\Cref{fig:rounding-error} shows that a low number of parameter bits and a high dynamic range (which we explain in \cref{sec:dr}) leads to a higher probability of changing the global optimum.
	
	In this work, we study the implications of solving \QUBO problems with limited parameter precision. 
	More specifically, we analyze the \emph{dynamic range}, \ie, the number of bits required to encode the \QUBO parameters faithfully:
	The smaller the dynamic range, the more robust the instance is against distortion. 
	We formalize the notion of optimum preservation between \QUBO instances based upon their set of minimizing bit vectors and explore to which extend parameters can be modified without changing said set.
	Finally, we introduce techniques to reduce the dynamic range of a given \QUBO instance based on theoretical bounds on the optimal energy value, which we demonstrate experimentally.
	The results demonstrate that the performance of quantum annealing hardware can be improved drastically when following our methodology.
	
	\subsection{Notation and Background}\label{sec:notation}
	We define $\mathbb{B}\defeq\lbrace 0,1\rbrace$ and $[n]\defeq\lbrace 1,\dots,n\rbrace$ for all $n\in\mathbb{N}$.
	The elements of $\BB$ are called \emph{bits}, and $\BB^n$ is the set of \emph{bit vectors} (or \emph{binary vectors}) of length $n$ for any $n\in\mathbb{N}$.
	
	\begin{definition}\label{def:fixed_subspace}
		Let $1\leq m\leq n$, $I\subseteq [n]$, $\vec{x}\in\mathbb{B}^n$, and $\vec{z}\in\mathbb{B}^m$. The $m$-dimensional sub-vector of $\vec{x}$ that contains only the variables indexed by $I$ is denoted via $\vec{x}_I\defeq(x_i)^\top_{i\in I}$.  The set of all $n$-dimensional bit vectors in which the variables indexed by $I$ are fixed to the values in $\vec{z}$ is denoted as $\mathbb{B}^n_{I\leftarrow \vec{z}}\defeq\lbrace \vec{x}': ~\vec{x}'\in\mathbb{B}^n, \vec{x}'_I=\vec{z}\rbrace$. 
	\end{definition}
	
	As a shorthand in later sections we may write $\mathbb{B}^n_{ij\leftarrow 10}$ instead of $\mathbb{B}^n_{\lbrace i,j\rbrace\leftarrow (1,0)^\top}$.
	Special binary vectors are $\vec 0^n\defeq (0,0,\dots,0)$ and $\vec 1^n\defeq (1,1,\dots,1)$. The superscript denotes their length, but we omit it when clear from context.
	
	Given any function $f:A\rightarrow B$ and a set $M\subseteq A$, we define $f(M)\defeq\set{f(a): ~a\in M}$ as the image set.
	Rounding a number $a\in\mathbb{R}$ to the nearest integer is denoted by $\nint{a}$.
	By convention, a number exactly halfway between two integers (with fractional part $0.5$) is rounded up.
	Additionally, we write $\nint{\mat{A}}$ to denote element-wise rounding for any real-valued vector or matrix $\mat{A}$.
	
	The object of study throughout this paper is the quadratic unconstrained binary optimization (\QUBO) problem, which is defined as follows:
	\begin{definition}\label{def:qubo}
		Let $\QUBOSET_n$ denote the set of all upper triangular matrices in $\mathbb{R}^{n\times n}$.
		A matrix $\mat{Q}\in\QUBOSET_n$ gives rise to a pseudo-boolean function $f_{\mat{Q}}:\mathbb{B}^n\rightarrow\mathbb{R}$ defined via \begin{equation*}\label{eq:qubo_def}
			f_{\mat{Q}}(\vec{x})\defeq \vec{x}^{\top}\mat{Q}\vec{x}=\sum_{i=1}^n\sum_{j=i}^n Q_{ij}x_ix_j\;.
		\end{equation*}
		This function is sometimes called \emph{energy}.
		The set of vectors which minimize $f_{\mat{Q}}$ is denoted by $\optset{\mat{Q}} := {\arg\min}_{\vec{x}} f_{\mat{Q}}(\vec{x})$. 
		Now, the
		\emph{quadratic unconstrained binary optimization} (\QUBO) problem is to determine $\vec{x}^*\in \optset{\mat{Q}}$. We drop the explicit dependence of $\optsetsymb$ on $\mat{Q}$ whenever it is clear from the context. 
	\end{definition}
	
	This problem is known to be \textsf{NP}-hard \cite{pardalos1992}, and its exact solution requires, in the worst case, an exhaustive search of an exponentially large candidate space. 
	A range of solution techniques has been developed over past decades.
	Some obtain the exact result \cite{narendra1977,rehfeldt2023} with worst-case exponential running time.
	Faster approximate techniques range from linear constraint solvers to heuristics such as simulated annealing \cite{kirkpatrick1983}, tabu search \cite{glover1998}, and genetic programming \cite{goldberg1987} -- see \cite{kochenberger2014} for a comprehensive overview.
	Perhaps most remarkably, \QUBO can be mapped to an Ising model \cite{brush1967} and solved through quantum annealing, which exploits quantum tunneling effects \cite{kadowaki1998}.
	
	\begin{lemma}\label{lemma:scaleinv}
		Let $\mat{Q}\in\QUBOSET_n$ and $\alpha>0$. The \QUBO problem is linear w.r.t. $\mat{Q}$, \ie, \begin{equation*}
			\alpha f_{\mat{Q}}(\vec{x}) = f_{\alpha \mat{Q}}(\vec{x}) \ \text{ and }\ f_{\mat{Q}+\mat{Q}'}(\vec{x}) = f_{\mat{Q}}(\vec{x}) + f_{\mat{Q}'}(\vec{x}), ~\forall \vec{x}\in\BB^n\;.
		\end{equation*}
	\end{lemma}
	
	The lemma tells us that the optimization problems described by $\alpha\mat{Q}$ are equivalent for all $\alpha>0$. As
	a consequence, the set of minimizing vectors remains unchanged, \ie, $\optset{\mat{Q}} = \optset{\alpha\mat{Q}}$, as well as the relative function value differences between all binary vectors.

	We will expand further on such types of relations between \QUBO instances in \cref{sec:preserving}. 
	
	\section{Parameter Precision}
	Even though the entries of a \QUBO matrix are real-valued in theory, on any real-world computing device there is a limit to the precision with which numbers can be represented.
	Typically, binary representations are used, where floating-point numbers or 2-complement integers can be represented with a fixed number of bits.
	For instance, a register of $B$ bits using 2-complement can represent all integers in $\lbrace -2^{B-1},\dots,2^{B-1}-1\rbrace$ (the first bit represents the sign).
	Values that have a fractional part must be rounded to the nearest integer in order to be represented in 2-complement.
	This leads to an important practical observation, which follows from \cref{lemma:scaleinv}:
	\begin{proposition}
		Any real-world computing device with finite floating-point precision can only solve \QUBO instances with integer parameters.
	\end{proposition}
	\begin{proof}
		Any number with finite floating-point precision is, in fact, rational, therefore we can assume $Q_{ij}=\frac{p_{ij}}{q_{ij}}$ for any $i\leq j\in [n]$ in a given \QUBO instance $\mat{Q}\in\QUBOSET_n$.
		Let $m\defeq \mathrm{lcm}\lbrace q_{ij}\rbrace_{i\leq j}$ the least common multiciple of all denominators, then $m\bQ$ has integer parameters.
	\end{proof}
	
	In this section we investigate how the precision of the parameter representation, specifically rounding, affects the energy landscape of $f_{\mat{Q}}$, and how we can evaluate the minimum number of bits necessary to represent the parameters of $\mat{Q}$ faithfully.
	Further we will define a relation between \QUBO instances $\mat{Q}$ and $\mat{Q}'$ that holds if some or all minimizing vectors are preserved between  $f_{\mat{Q}}$ and $f_{\mat{Q}'}$.
	Using this concept, we show that there is a theoretical lower bound on a scaling factor that preserves at least one of the optima after rounding the parameters.
	
	\subsection{Preserving Optima}\label{sec:preserving}
	Every \QUBO instance has at least one binary vector with minimum energy, which is the global optimum of the optimization problem. 
	As known from Lemma~\ref{lemma:scaleinv}, scaling a \QUBO instance with a positive factor $\alpha$ does not change the optima, therefore $\optset{\mat Q}=\optset{\alpha\mat Q}$.
	Note that $\optset{\mat Q}$ cannot be the empty set, as in a non-empty finite set of real numbers a smallest element always exists.
	
	Having a choice between multiple optima has several benefits.
	For one, when there is a variety of equally good solutions, we can choose one according to other criteria that are not encoded in the optimization problem (\eg, the solution with fewest 1-bits).
	Moreover, certain iterative optimization strategies may converge more quickly when there are multiple global optima scattered through the search space, as the distance to the nearest optimum across all binary vectors decreases.
	However, when we are interested in just \emph{any} optimal solution, we have to accept that, when modifying the parameters of $\mat{Q}$, some optima may get lost.
	To this end, we define the notion of \emph{optimum inclusion} on \QUBO instances:
	\begin{definition}
		Let $\mat Q,\mat Q'\in\QUBOSET_n$.
		We say that $\mat Q$ \emph{includes the optima of} $\mat Q'$, written as $\mat Q\optleq\mat Q'$, when the set of optima of $\mat Q$ is a subset of the optima of $\mat Q'$: \begin{equation*}
			\mat Q\optleq\mat Q' ~\Longleftrightarrow ~\optset{\mat Q}\subseteq \optset{\mat Q'}
		\end{equation*}
	\end{definition}
	
	The relation $\optleq$ induces a preorder on $\QUBOSET_n$.
	It is not antisymmetric, as $\mat Q\optleq\mat Q'$ and $\mat Q'\optleq\mat Q$ does not imply $\mat Q=\mat Q'$, but it is reflexive and transitive.
	Informally, if $\mat Q\optleq\mat Q'$, we know that $\mat Q$ and $\mat Q'$ share at least one global optimum, and that $\mat Q$ has no optima that $\mat Q'$ does not have as well.
	Therefore, a binary vector that minimizes $f_{\mat Q}$ is guaranteed to also minimize $f_{\mat Q'}$.
	
	\subsection{Dynamic Range}
	\label{sec:dr}
	To make the degree of precision we need to faithfully represent \QUBO parameters measurable, we borrow the notion of \emph{dynamic range} from signal processing.
	\begin{definition}
		Let $X\subset\mathbb{R}$ be a finite set.
		First, we define the set of absolute differences between all elements as $\diffset{X}\defeq \lbrace \abs{x-y}:~x,y\in X, \,x\neq y\rbrace$, and we write $\minD{X}\defeq\min\diffset{X}$ and $\maxD{X}\defeq\max\diffset{X}$.
		The \emph{dynamic range} (DR) of $X$ is defined as \begin{equation}
			\DR(X) \defeq \log_2\biggl(\frac{\maxD{X}}{\minD{X}}\biggr) \label{eq:dynamic_range}
		\end{equation}
		Its unit is \emph{bits}.
		The DR of a \QUBO matrix $\mat{Q}$ is defined as the DR of the set of its entries:
		\begin{align*}
			\DR(\mat{Q}) &\defeq \DR(\elemset{\mat Q}),\ &\text{where }~\elemset{\mat Q} &\defeq \lbrace Q_{ij}:~i,j\in[n]\rbrace\;.
		\end{align*}
		Note that always $0\in\elemset{\mat Q}$, because $\mat Q$ is triangular, thus $Q_{ij}=0$ when $i>j$.
		In the following, we will simply write $\diffset{\mat Q}$, $\minD{\mat Q}$ and $\maxD{\mat Q}$ instead of $\diffset{\elemset{\mat Q}}$, $\minD{\elemset{\mat Q}}$ and $\maxD{\elemset{\mat Q}}$.
	\end{definition}
	
	A large dynamic range implies that $\mat{Q}$ needs many bits to represent all parameters faithfully in binary, because its parameters both cover a large value range and require small gradations.
	Note that scaling has no effect on dynamic range.
	
	\begin{example}
		Consider the following $2\times 2$ upper triangular matrices:
		\begin{align*}%\label{eq:simplereduction}
			\mat{Q} &= \begin{bmatrix}-1 &14380 \\ &-2\end{bmatrix} &\mat{Q}' &=\begin{bmatrix}-1 &3 \\ &-2\end{bmatrix}
		\end{align*}
		Matrix $\mat{Q}$ defines the \QUBO problem $f_{\mat{Q}}(\bx)=-x_1+14380x_1x_2-2x_2$.
		The global optimum is $\vec{x}^*=(0,1)^\top$ with value $f_{\mat{Q}}(\vec{x}^*)=-2$.
		As we see at a glance, the value $Q_{12}$ is positive and very large, acting as a penalty weight between bits $x_1$ and $x_2$.
		However, a much smaller value has the same effect, as we see with $\mat{Q}'$, which is identical to $\mat{Q}$ except for $Q'_{12}=3$.
		The global optimum of $f_{\mat{Q}'}$ is still $\vec{x}^*=(0,1)^{\top}$ with $f_{\mat{Q}'}(\vec{x}^*)=-2$, therefore $\mat Q\optleq\mat Q'$.
		When we compare the dynamic ranges of $\mat{Q}$ and $\mat{Q}'$, we find that \begin{align*}
			\DR(\mat{Q})&=\log_{2}\left(\frac{14380-(-2)}{0-(-1)}\right) &\DR(\mat{Q}')&=\log_{2}\left(\frac{3-(-2)}{0-(-1)}\right) \\
			&\approx 13.812 &&\approx 2.322,
		\end{align*}
		which is a tremendous reduction:
		While we need 14 bits to encode the elements of $\mat{Q}$, we only need 3 for $\mat{Q}'$.
		This demonstrates that, in principle, it is possible to reduce the dynamic range while preserving the minimizing vectors of the \QUBO problem.
	\end{example}
	
	One way to enforce a reduction of dynamic range is scaling and rounding, which sets the smallest meaningful difference between values to 1.
	It is easy to see that this gives us an upper bound on the DR:
	\begin{proposition}\label{prop:roundingdr}
		Let $\mat{Q}\in\QUBOSET_n$; w.l.o.g. we assume that $\mat Q$ is scaled such that $\maxD{\mat Q}=1$.
		For any $\alpha>0$ the DR of $\nint{\alpha\mat Q}$ is bounded above by
		\begin{align*}
			\DR(\nint{\alpha\mat Q})&\leq\log_2(1+\alpha)\;.
		\end{align*}
	\end{proposition}
	\begin{proof}
		Due to the normalization we have $\DR(\alpha\mat Q)=\log_2(\alpha/\beta)$ with $\beta\le\alpha$.
		By rounding we enforce $\minD{\nint{\alpha\mat Q}}\geq 1$, therefore \begin{align*}
			\DR(\nint{\alpha\mat Q}) &= \log_2\biggl(\frac{\epsilon+\alpha}{\epsilon'+1}\biggr) &&(-1\leq\epsilon<1)\\
			&\leq \log_2\biggl(\frac{1+\alpha}{\epsilon'+1}\biggr) &&(0\leq\epsilon')\\
			&\leq \log_2(1+\alpha)\;.
		\end{align*}
	\end{proof}
	
	Rounding a \QUBO instance $\mat Q$ does not, in general, preserve $\optset{\mat Q}$, as it skews the values of $f_{\mat Q}$.
	We can model this perturbation caused by rounding after scaling with some $\alpha>0$ as \begin{align*}
		f_{\nint{\alpha\bQ}}(\bx) &= \bx^{\top}\nint{\alpha\bQ}\bx \\
		&= \bx^{\top}\bigl(\alpha\bQ+\bE(\bQ,\alpha)\bigr)\bx \\
		&= \alpha\bx^{\top}\bQ\bx+\bx^{\top}\bE(\bQ,\alpha)\bx \\
		&= \alpha f_{\bQ}(\bx)+f_{\bE(\bQ,\alpha)}(\bx)\;.
	\end{align*}%
	Here, $\bE(\bQ,\alpha)\in(-\frac{1}{2},\frac{1}{2}]^{n\times n}$ denotes the matrix of difference between the real values and their nearest integers in $\mat Q$ after scaling with $\alpha>0$.
	When we reverse the scaling after rounding, the error on each entry in $\mat Q$ is bounded in $(-\frac{1}{2\alpha},\frac{1}{2\alpha}]$, and the error on the total function value of any bit vector $\bx$ is exactly \begin{equation}\label{eq:roundingerror}
		\epsilon_{\mat Q,\alpha}(\vec x)\defeq f_{\bE(\bQ,\alpha)}(\bx)/\alpha\;.
	\end{equation}
	This matrix sum representation of rounding bridges the gap to the precision of the D-Wave quantum annealers, where the ICEs can be modeled by adding random noise to the parameters of the underlying Ising model.
	
	\subsection{Optimal Rounding}
	The question remains how to ensure that the rounding error does not lead to false optima.
	As \cref{eq:roundingerror} shows, the magnitude of the error generally decreases as $\alpha$ increases, because $\abs{\mat{E}(\mat Q, \alpha)_{ij}}\leq 0.5,~\forall i,j\in [n]$, regardless of $\mat Q$ and $\alpha$.
	This allows us to bound the error between $f_{\nint{\alpha\mat Q}/\alpha}$ and $f_{\mat Q}$ within an interval of $\pm\frac{C}{\alpha}$ for some constant $C>0$.
	If we choose $\alpha$ such that the rounding errors cannot bridge the gap between the lowest and second to lowest value of $f_{\mat Q}$, we can guarantee that no non-optimal vector's value is rounded down to the optima's value.
	\begin{definition}\label{def:spectralgap}
		Let $\mat{Q}\in\QUBOSET_n$, and let $y_1 \defeq \min f_{\mat{Q}}(\BB^n)$ its lowest and $y_2 \defeq \min(f_{\mat{Q}}(\BB^n)\backslash\set{y_1})$ its second-to lowest value w.r.t. $f_{\mat Q}$, then the \emph{spectral gap} $\gamma_{\mat{Q}}$ is defined as \begin{equation*}
			\gamma_{\mat{Q}} \defeq y_2-y_1\;.
		\end{equation*}
	\end{definition}
	
	Originally, the spectral gap is defined as the difference between the lowest and second to lowest eigenvalue of a Hamiltonian operator in physics, however the concept can be extended to any optimization problem with a real-valued objective function.
	
	In general, computing the spectral gap is at least as difficult as solving the \QUBO problem itself, i.e., intractable for large $n$.
	However, it gives us a theoretical lower bound for $\alpha$ that allows for optimum-preserving rounding:
	\begin{theorem}\label{theorem:gamma}
		Let $\mat{Q}\in\QUBOSET_n$ and $\gamma_{\mat{Q}}$ its spectral gap.
		Then \begin{align*}
			\forall \alpha\geq\alpha^*: ~\nint{\alpha\mat{Q}}\optleq\mat{Q} &&\text{where }\alpha^*\defeq \frac{n^2+n}{4\gamma_{\mat Q}}\;.
		\end{align*}
	\end{theorem}%
	\begin{proof}
		Recall \cref{eq:roundingerror}.
		Each entry of $\mat{E}(\mat Q,\alpha)$ is bounded in $(-\frac{1}{2},\frac{1}{2}]$.
		It is clear that the worst-case rounding error is proportional to the norm of the binary vector, as $f_{\mat Q}(\vec x)$ contains more summands than $f_{\mat Q}(\vec x')$ if $\norm{\vec x}_1>\norm{\vec x'}_1$.
		Accordingly, the overall worst-case rounding error would occur if each entry of $\mat E$ was $+\frac{1}{2}$ or close to $-\frac{1}{2}$ respectively, one of the optimal or second to optimal vectors was $\bm 1$ and the other $\bm 0$.
		Assume that $\bm 1$ is a global minimum and $\bm 0$ has the second to minimal value, then we find 
		\begin{equation*}
			-\frac{n^2+n}{4}=-\frac12\frac{n(n+1)}{2}<\bm 1^{\top}\mat E(\mat Q,\alpha)\bm 1\leq \frac12\frac{n(n+1)}{2}= \frac{n^2+n}{4}\;,
		\end{equation*}%
		and $f_{\nint{\alpha\mat Q}}(\bm 0)=f_{\mat Q}(\bm 0)=0$.
		From \cref{def:spectralgap} we find that $\gamma_{\alpha\mat{Q}}=\alpha\gamma_{\mat{Q}}$ for any $\alpha>0$.
		To ensure that the combined rounding error cannot cause $f_{\nint{\alpha\mat Q}}(\bm 0) -f_{\nint{\alpha\mat Q}}(\bm 1)\le 0$, we need 
		\begin{align*}
			0&\overset{!}{\le} f_{\nint{\alpha\bQ}}(\bm 0) - f_{\nint{\alpha\bQ}}(\bm 1) 
			= f_{\alpha\bQ}(\bm 0)-f_{\alpha\bQ}(\bm 1)- \bm 1^{\top}\mat E(\mat Q,\alpha)\bm 1 \\
			&=\alpha\gamma_{\mat Q}-\bm 1^{\top}\mat E(\mat Q,\alpha)\bm 1 
			\le\alpha\gamma_{\mat Q}-\frac{n^2+n}{4} \\
			\Leftrightarrow ~\alpha &\overset{!}{\ge}\frac{n^2+n}{4\gamma_{\mat Q}}\;.
		\end{align*}
	\end{proof}
	
	\section{Parameter Compression}
	As we have seen in \cref{prop:roundingdr}, scaling and rounding reduces the dynamic range.
	However, in general, reducing the DR leads to coarser energy gradations.
	E.g., a \QUBO instance where parameters are encoded with only 2 bits can only have values in $\singleton{-2,-1,0,1}$, which may be insufficient to accurately preserve the value function and, consequently, the minimizing binary vectors.
	In this section we develop strategies to balance these competing objectives and reduce the DR while keeping (some of) the optimal vectors intact.
	To this end, we modify parameter values while trying to stay within bounds which guarantee that a minimal solution stays minimal, and a non-minimal solution does not become minimal.
	
	For now, assume a \QUBO instance $\mat{Q}\in\QUBOSET_n$ has a unique global minimum $\vec{x}^*$ of value $y^*\defeq f_{\mat{Q}}(\vec{x}^*)$.
	Our two objectives can be formulated as a constrained optimization problem, where we want to find a matrix $\mat{A}\in\QUBOSET_n$ such that
	\begin{align*}
		\argmin_{\mat{A}\in\QUBOSET_n}\quad &\DR(\mat{Q}+\mat{A})\\%\label{eq:dr_reduce_change}\\
		\text{s.t.}\quad &\mat{Q}+\mat{A}\optleq \mat{Q}\;.%\label{eq:preserv_opt_change}
	\end{align*}
	To approach this problem, we update the parameters $Q_{kl}\mapsto Q_{kl}+\wut$ sequentially.
	Deciding which \QUBO parameter to update and how to choose the right value for $\wut$ will be the focus of the following subsections.
	For now, let the indices $k,l\in [n], k\le l$ of some parameter within $\mat Q$ be arbitrary but fixed.
	For conciseness, we write $\vec x^*_{ab}$ and $y_{ab}$ instead of $x^*_{kl\leftarrow ab}$ and $y_{kl\leftarrow ab}$.
	
	\subsection{Bounding the Optimal Value}
	Recall our definition of $\mathbb{B}^n_{I\leftarrow \vec{z}}$ from \cref{sec:notation}.
	Fixing one or more bits in a binary vector to constants induces subspaces of $\BB^n$, one for each possible assignment of variables indexed by $I$, which is $2^{\abs I}$ in total.
	Each subspace has its own set of minimizing binary vectors w.r.t. $f_{\mat{Q}}$.
	\begin{definition}\label{def:subspace}
		Given $\mat Q\in\QUBOSET_n$, indices $I\subseteq [n]$ and an assignment $\vec z\in\BB^{\abs I}$, the subspace optima are defined as \begin{equation*}
			\suboptset{I\leftarrow\vec z}{\mat Q}\defeq\lbrace\vec x^*: ~\vec x^*\in\BB^n_{I\leftarrow\vec z},\,f_{\mat Q}(\vec x^*)\leq f_{\mat Q}(\vec x) ~\forall\vec x\in\BB^n_{I\leftarrow\vec z}\rbrace\;.
		\end{equation*}
	\end{definition}
	
	Note that $\suboptset{I\leftarrow\vec z}{\mat Q}\cap\suboptset{I\leftarrow\vec z'}{\mat Q}=\emptyset$ for any $\vec z\neq\vec z'$.
	Therefore we can choose an arbitrary (but fixed) element $\vec x^*_{ab}\defeq\suboptset{kl\leftarrow ab}{\mat Q}$ as a representative for all $ab\in\BB^2$.
	Their respective values are denoted by $y^*_{ab}\defeq f_{\mat{Q}}(\vec{x}^*_{ab})$.
	Naturally, $y^*_{ab}$ are just as hard to compute as solving \QUBO itself.
	Therefore, we work with upper and lower bounds on the true values, which are much easier to compute.
	We will use them to determine the update parameter $\wut$.
	\begin{definition}\label{def:bounds}
		Upper (lower) bounds for $y^*_{ab}$ are denoted by $\hat y_{ab}$ ($\check y_{ab}$), such that
		\begin{equation*}
			\check y_{ab}\le y^*_{ab}\le\hat y_{ab},~\forall ab\in\BB^2\;.
		\end{equation*}
		Further, let
		\begin{align*}
			\wuo^- &\defeq \min \{0,\min\lbrace\hat{y}_{00},\hat{y}_{01},\hat{y}_{10}\rbrace-\check{y}_{11}\} \;,\\
			\wuo^+ &\defeq \max \{0, \min\lbrace\check{y}_{00},\check{y}_{01},\check{y}_{10}\rbrace-\hat{y}_{11}\}\;,
		\end{align*}
		if $k\not=l$. Otherwise, when $k=l$, let $\wuo^-=\min\{0,\hat{y}_{00} - \check{y}_{11}\}$ and $\wuo^+ = \max\{0,\check{y}_{00} - \hat{y}_{11}\}$. 
	\end{definition}
	\begin{theorem}\label{theorem:optimum_bound}
		Let $\mat Q\in\QUBOSET_n$, and assume we perform an update $Q_{kl}\mapsto Q_{kl}+\wut$.
		Given $\wuo^-$ and $\wuo^+$ as defined in \cref{def:bounds}, then an optimum is preserved as long as
		\begin{align}
			\wuo^-\le\wut \le \wuo^+\;.\label{eq:bound_opt}
		\end{align}
	\end{theorem}
	\begin{proof}
		We focus on the case $k\not=l$, the case $k=l$ is analogous. 
		A global minimum $y^*$ must be equal to exactly one of the four subspaces' optima $y^*_{00},y^*_{01},y^*_{10},y^*_{11}$.
		Notice that changing $Q_{kl}$ by $\wut$ affects only $y^*_{11}$.
		Assume $\vec{x}^*\neq \vec{x}^*_{11}$, then an optimum is preserved if
		\begin{align}
			&y^*_{11}+\wut\ge y^* \nonumber\\
			\Leftrightarrow\ &y^*_{11}+\wut\ge\min\{y^*_{00},y^*_{01},y^*_{10}\} \nonumber\\
			\Leftarrow\ & \check{y}_{11}+\wut\ge\min\{\hat{y}_{00},\hat{y}_{01},\hat{y}_{10}\}\;.\label{eq:opt_bound_proof}
		\end{align}
		Furthermore, $\wut$ can take any positive value, since $y^*_{11}> y^*$. 
		Combining this observation with \cref{eq:opt_bound_proof}, we end up with the lower bound	\begin{equation}
			\wut \ge \min \{0,\min\lbrace \hat{y}_{00},\hat{y}_{01},\hat{y}_{10}\rbrace-\check{y}_{11}\}=\wuo^-\;. \label{eq:bound_opt_lower}
		\end{equation}
		If $\vec{x}^*= \vec{x}^*_{11}$, we can similarly deduce an upper bound
		\begin{equation}
			\wut \le \max \left\lbrace0,\min\lbrace \check{y}_{00},\check{y}_{01},\check{y}_{10}\rbrace-\hat{y}_{11}\right\rbrace=\wuo^+\;.
			\label{eq:bound_opt_upper}
		\end{equation}
		Combining \cref{eq:bound_opt_lower,eq:bound_opt_upper} we obtain \cref{eq:bound_opt}.
	\end{proof}
	
	\begin{figure}
		\centering
		\begin{subfigure}{.4\textwidth}
			\centering
			\begin{tikzpicture}[scale=.9]
				\interv{0.5}{.5}{1.5}{$\check y_{11}$}{$\hat y_{11}$}
				\interv{-2}{1}{1}{$\check y_{10}$}{$\hat y_{10}$}
				\interv{-3}{1.5}{2.5}{$\check y_{01}$}{$\hat y_{01}$}
				\interv{-3.5}{2}{2.8}{$\check y_{00}$}{$\hat y_{00}$}
				\draw[dotted] (0.5,0.3) -- (0.5,2.2);
				\draw[dotted] (-1,0.3) -- (-1,2.2);
				\node[anchor=north] at (0.5,.2) {$\check{y}_{11}$};
				\node[anchor=north] at (-1,.2) {$\min\singleton{\hat y_{ab}}$};
				\draw[-latex] (-4,3) -- (2.2,3) node[anchor=west] {$\mathbb{R}$};
				\intervb{-3.5}{2.7}{2.5}{}{}
			\end{tikzpicture}
			\caption{$\check{y}_{11}\ge\min\{\hat{y}_{00},\hat{y}_{01},\hat{y}_{10}\}$ implies that $\vec x^*\neq\vec x^*_{11}$.}
		\end{subfigure}%
		\hfill
		\begin{subfigure}{0.4\textwidth}
			\centering
			\begin{tikzpicture}[scale=.9]
				\interv{-2.5}{.5}{1.5}{$\check y_{11}$}{$\hat y_{11}$}
				\interv{0.5}{1}{1}{$\check y_{10}$}{$\hat y_{10}$}
				\interv{0.3}{1.5}{2}{$\check y_{01}$}{$\hat y_{01}$}
				\interv{0.8}{2}{0.8}{$\check y_{00}$}{$\hat y_{00}$}
				\draw[dotted] (-1,0.3) -- (-1,2.2);
				\draw[dotted] (.3,.3) -- (.3,2.2);
				\node[anchor=north] at (-1,.2) {$\hat{y}_{11}$};
				\node[anchor=north] at (.3,.2) {$\min\singleton{\check y_{ab}}$};
				\draw[-latex] (-3,3) -- (2.2,3) node[anchor=west] {$\mathbb{R}$};
				\intervb{-2.5}{2.7}{1.5}{}{}
			\end{tikzpicture}
			\caption{$\hat{y}_{11}\le\min\{\check{y}_{00},\check{y}_{01},\check{y}_{10}\}$ implies that $\vec x^*=\vec x^*_{11}$.}
		\end{subfigure}
		\caption{Illustration of \cref{prop:check_optimum}, the green bars indicate the interval the global optimum must fall into; when the lower bound for a subspace $\BB^n_{kl\rightarrow ab}$ is greater than an upper bound of any other subspace, we can conclude that $\vec x^*_{ab}$ cannot be optimal (a).
			On the other hand, when an the upper bound for a subspace $\BB^n_{kl\rightarrow ab}$ is lower than the lower bounds of all other subspaces, we can conclude that $\vec x^*_{ab}$ is optimal (b).
			For above example, we set $ab=11$.}
		\label{fig:optbounds}
	\end{figure}
	
	\cref{eq:bound_opt} uses bounds ($\check y_{ab}$ and $\hat y_{ab}$) on the true optima $y^*_{ab}$ to give us an interval for $\wut$ if we want to preserve an optimum.
	These bounds can also be used for determining optimality of $\vec{x}^*_{ab}$.
	\begin{proposition}\label{prop:check_optimum}
		The following implications hold:
		\begin{align}
			\check{y}_{ab}>\min\left( \{\hat{y}_{00},\hat{y}_{01},\hat{y}_{10},\hat{y}_{11}\}\setminus \{\hat{y}_{ab}\}\right) &\Rightarrow \vec{x}^*\neq\vec{x}^*_{ab} \label{eq:check_suboptimum} \\
			\hat{y}_{ab}<\min\left( \{\check{y}_{00},\check{y}_{01},\check{y}_{10},\check{y}_{11}\}\setminus\{\check{y}_{ab}\}\right)  &\Rightarrow \vec{x}^*=\vec{x}^*_{ab} \label{eq:check_optimum}
		\end{align}
	\end{proposition}
	\begin{proof}
		We focus on the case $k\not=l$, the case $k=l$ is analogous. Assume that \cref{eq:check_suboptimum} holds, \ie,
		\begin{align*}
			&\ \check{y}_{ab}>\min\left( \{\hat{y}_{00},\hat{y}_{01},\hat{y}_{10},\hat{y}_{11}\}\setminus \{\hat{y}_{ab}\}\right) \\
			\Rightarrow&\ y^*_{ab} > \min\left( \{y^*_{00},y^*_{01},y^*_{10},y^*_{11}\}\setminus \{y^*_{ab}\}\right) \\
			\Leftrightarrow&\ y^*_{ab} > y^* ~\Leftrightarrow ~\vec{x}^*\neq\vec{x}^*_{ab}\;.
		\end{align*}
		The result in \cref{eq:check_optimum} follows analogously.
		See \cref{fig:optbounds} for a visualization.
	\end{proof}
	If we find the inequality in \cref{eq:check_optimum} to be true for some two-bit assignment $ab$, the dimension of the \QUBO search space can be reduced by fixing $x_k=a$ and $x_l=b$.
	Similar reduction techniques can be found in \cite{boros2006, lewis2017, glover2018}.
	Knowing that $\vec{x}^*\neq\vec{x}^*_{11}$ we can get rid of the upper bound \cref{eq:bound_opt} (cf. proof of \cref{theorem:optimum_bound}).
	
	The questions remains how to obtain lower and upper bounds on $y^*_{ab}$.
	Simple, but weak bounds can be computed in $\Theta(n^2)$.
	\begin{proposition}
		Let $k,l$ be two indices with $1\leq k\le l\leq n$ and $ab\in\BB^2$ a variable assignment.
		An upper bound for $y^*_{ab}$ is given by \begin{equation*}
			y^*_{ab}\leq f_{\mat Q}(\vec 0_{kl\leftarrow ab})\;.
		\end{equation*}
	\end{proposition}
	\noindent
	To obtain better upper bounds, we can invest more computational effort and, for instance, perform a local search in the space $\BB^n_{kl\leftarrow ab}$, or perform rejection sampling, and record the lowest observed value of $f_{\mat Q}$ s.t. $\vec x_{kl}=ab$.
	
	In order to derive simple lower bounds we can take only the negative entries of $\mat Q$ and compute the lowest possible sum that can be formed from them.
	\begin{proposition}
		Let $k,l$ and $ab$ as before.
		Define $\mat Q^-$ such that $Q^-_{ij}=\min\singleton{0,Q_{ij}}$, \ie, the matrix containing only the negative values of $\mat Q$. 
		Then a lower bound for $y^*_{ab}$ is given by \begin{equation*}
			y^*_{ab}\geq f_{\mat Q^-}(\vec 1_{kl\leftarrow ab})\;.
		\end{equation*}
	\end{proposition}
	\noindent
	The lower bound can be improved by exploiting roof duality \cite{boros2008, hammer1984} at the cost of higher (yet polynomial) computational complexity.
	
	\subsection{Reducing the Dynamic Range}\label{sec:reduce_methods}
	We have established intervals within which the parameters of a \QUBO problem can be modified while an optimum is preserved.
	The question remains how to choose the values in a way that reduces the dynamic range.
	In this subsection we show multiple approaches to achieve this goal.
	
	Let $m\defeq n^2$ be the number of entries of an $n$ by $n$ square matrix.
	For any $\mat{Q}\in\QUBOSET_n$ there is an ordering $\pi:[m]\rightarrow [n]^2$ of values in $\mat{Q}$ such that 
	\begin{equation*}
		\sortQ{i}\leq \sortQ{i+1}, ~\sortQ{i}\equiv Q_{\pi(i)},~\forall i\in[m]
	\end{equation*}
	Using this notation, $\sortQ{1}=\min\,\elemset{\mat Q}$, $\sortQ{m}=\max\,\elemset{\mat Q}$, and further $\maxD{\mat Q}=\sortQ{m}-\sortQ{1}$ and $\exists j\in[m-1]:~\minD{\mat Q}=q_{j+1}-q_j$.
	Note that about half of all $q_i$ are 0, as $\mat Q$ is upper triangular.
	
	\begin{theorem}\label{theorem:dynamic_range}
		Let $\mat{Q}\in\QUBOSET_n$ and $\pi(\ell)=(k,l)$ with $k\leq l$ and $Q_{kl}\neq 0$.
		When adding a value $\wut\in\mathbb R$ to $Q_{kl}$, the dynamic range does not increase, \ie, $\DR(\mat{Q})\ge\DR(\mat{Q}+\wut\vec{e}_k\vec{e}_l^{\top})$, if the following two conditions hold:
		\begin{enumerate}
			\item $\wut$ is bounded: \begin{equation}
				\underbrace{\sortQ{1}-\sortQ{\ell} +\delta_{m\ell}\left(\sortQ{m-1}- \sortQ{m}\right)- \drs{\ell}}_{\eqdef\wud^-}\le\wut\le \underbrace{\sortQ{m}-\sortQ{\ell} +\delta_{1\ell}\left(\sortQ{2}- \sortQ{1}\right)+ \drs{\ell}}_{\eqdef\wud^+}\label{eq:dr_bounds}\;,
			\end{equation}
			\item $\wut$ does not decrease the minimal parameter distance:
			\begin{align}
				&&\left|\sortQ{\ell}+\wut-\sortQ{i}\right| &\ge \minD{\mat Q},\ &&\forall i\in [m]\setminus\{\ell\}\label{eq:dr_distance} \\
				\vee&&\sortQ{\ell}+\wut &=\sortQ{i}, &&\exists i\in[m]\;.\label{eq:dr_landing}
			\end{align}
		\end{enumerate}
		Here, $\delta_{\cdot}$ is the Kronecker delta with $\delta_{uv}\defeq 1$ if $u=v$, else $0$, and $\drs{\ell}$ is defined as
		\begin{equation*}
			\drs{\ell}\defeq
			\maxD{\mat Q}\left(\frac{\minD{\singleton{q_u: ~u\in [m]\backslash\singleton{\ell}}}}{\minD{\mat Q}}-1\right)\;.
		\end{equation*}
	\end{theorem}
	\begin{proof}
		We only consider an increase of the \QUBO parameter $Q_{kl}$, \ie $\wut>0$, since the results for decreasing $Q_{kl}$ can be deduced analogously.
		Firstly, consider the parameters $\sortQ{\ell}>\sortQ{1}$.
		If 
		\begin{align}
			\wut\le \sortQ{m} - \sortQ{\ell}\;, \label{eq:incr_bound_max} 
		\end{align}
		then $\maxD{\mat Q}$ is not increased and thus to avoid an increase of the $\DR$, $\minD{\mat Q}$ should not be decreased (see \cref{eq:dynamic_range}). 
		This can be achieved by maintaining a distance of at least $\minD{\mat Q}$ to all other \QUBO parameters, \ie,
		\begin{equation}
			\left|\sortQ{\ell} + \wut - \sortQ{i}\right|\ge \minD{\mat Q},~\forall i\in [m]\setminus\{\ell\}\;, \label{eq:avoid_min_decr}
		\end{equation}
		or ``landing'' on an already existing \QUBO parameter, \ie,
		\begin{equation}
			\sortQ{\ell} + \wut=\sortQ{i},~\exists i\in[m]\;. \label{eq:landing}
		\end{equation}
		If the current maximum value is overshot, \ie,
		\begin{align*}
			\wut>\sortQ{m}-\sortQ{\ell}\;,
		\end{align*}
		$\maxD{\mat Q}$ is increased and thus to reduce the $\DR$, $\minD{\mat Q}$ has also to be increased.
		This can only be the case if $\sortQ{\ell}$ is unique and is part of the minimum distance, \ie, $\minD{\mat Q\setminus\{\sortQ{\ell}\}}>\minD{\mat Q}$.
		The change $\wut$ can then be bounded by
		\begin{align*}
			\DR\left(\mat{Q}\right)&\ge \DR\left(\mat{Q}+\wut\vec{e}_k\vec{e}_l^{\top}\right) \\
			\Leftrightarrow\frac{\maxD{\mat Q}}{\minD{\mat Q}}&\ge\frac{\maxD{\mat Q}+\sortQ{\ell}+\wut-\sortQ{m}}{\min\singleton{\minD{\singleton{q_u: ~u\in [m]\backslash\singleton{\ell}}},\minD{\mat Q}+\wut}} \\
			&\ge\frac{\maxD{\mat Q}+\sortQ{\ell}+\wut-\sortQ{m}}{\minD{\singleton{q_u: ~u\in [m]\backslash\singleton{\ell}}}} \\
			\Leftrightarrow \wut &\le \sortQ{m}-\sortQ{\ell} + \maxD{\mat Q}\left(\frac{\minD{\singleton{q_u: ~u\in [m]\backslash\singleton{\ell}}}}{\minD{\mat Q}}-1\right)\;.  \label{eq:dr_decr_overshoot_inner}
		\end{align*}
		
		Secondly, consider $\ell=1$. If the smallest value $\sortQ{1}$ is not unique, we can also deduce bounds \cref{eq:incr_bound_max,eq:avoid_min_decr,eq:landing}. 
		On the other hand, when the smallest value is unique ($\sortQ{2}-\sortQ{1}>0$) the increase $\sortQ{1}+\wut>\sortQ{m}$ does not necessarily increase $\maxD{\mat Q}$. 
		If $\wut >\sortQ{2}-\sortQ{1}$, $\sortQ{1} + \wut$ is not the minimum value anymore but $\sortQ{2}$ is.
		Thus, the difference $\sortQ{2}-\sortQ{1}$ can be added to the bound in \cref{eq:incr_bound_max}
		\begin{align*}
			\wut&\le\left(\sortQ{m} - \sortQ{1}\right) + \left(\sortQ{2}- \sortQ{1}\right)=\sortQ{m} - 2\sortQ{1} + \sortQ{2}\;.
		\end{align*}
		Additionally, if $\sortQ{1}$ is part of the unique minimum distance, we can add $\sortQ{2}-\sortQ{1}$ to the bound in \cref{eq:dr_decr_overshoot_inner}
		\begin{align}
			\wut &\le \sortQ{m}-\sortQ{1} + \sortQ{2}-\sortQ{1}+ \maxD{\mat Q}\left(\frac{\minD{\singleton{q_u: ~u\in [m]\backslash\singleton{\ell}}}}{\minD{\mat Q}}-1\right)\;.%\label{eq:dr_decr_overshoot_smallest}
		\end{align}
		Similar bounds can be obtained for a negative change $\wut<0$.
	\end{proof}
	
	\cref{eq:dr_bounds,eq:dr_distance,eq:dr_landing} give us very loose bounds on the \QUBO parameter changes. 
	To determine the exact value, we present a few heuristic approaches.

	\subsubsection{A Greedy Strategy}
	The first heuristic is a greedy strategy which we denote by $\GRE$.
	The \QUBO parameter $Q_{kl}$ is increased if $\sortQ{\ell}<0$ and decreased otherwise, where $\pi(\ell)=(k,l)$.
	For increasing (decreasing) $Q_{kl}$ we choose $\wuh^{\GRE}$ maximally (minimally), \ie, $\wuh^{\GRE}=\wud^+$ ($\wuh^{\GRE}=\wud^-$). 
	If the updated \QUBO parameter lays too close to some other parameter, \ie, 
	\begin{equation*}
		\left|\sortQ{\ell} + \wud^{\pm} - \sortQ{i}\right|\ge \minD{\mat Q},\ \exists i\in [m]\setminus\{\ell\}\;,
	\end{equation*}
	we would set it equal to the next smaller (larger) \QUBO parameter, that is,
	\begin{equation}
		\sortQ{\ell} + \wuh^{\GRE}=\sortQ{i}, ~\sortQ{j}\le \sortQ{\ell} + \wud^{+},~\forall j\le i ~\left(\sortQ{j}\ge \sortQ{\ell} + \wud^{-},~\forall j\ge i\right)\;. \label{eq:gre_next}
	\end{equation}
	Again, recall that there is always a $q_u=0$ for some $u\in[m]$, and thus we may set parameters to $0$.
	For certain target platforms, such as quantum annealers, this is particularly beneficial, as setting a parameter to $0$ allows to discard the coupling between the qubits indexed by $k$ and $l$, which saves hardware resources.
	As an alternative version to \cref{eq:gre_next}, we choose $\wuh^{\GRE_0}$ such that
	\begin{equation*}
		\sortQ{\ell} + \wuh^{\GRE_0}=0, ~0\le \sortQ{\ell} + \wud^{+}, ~\left(0\ge \sortQ{\ell} + \wud^{-}\right)\;. 
	\end{equation*}
	We henceforth call this alternative version $\GRE_0$. 
	
	\subsubsection{Maintaining the Parameter Ordering}
	With the preceding methods, we allowed parameters to cross over each other, changing their ordering.
	However, another heuristic approach is to maintain the ordering of the elements in $\elemset{\mat Q}$, which should intuitively help to preserve the optimum.
	This heuristic is denoted by $\MOR$.
	For this, define bounds on a certain \QUBO parameter $q_{\ell}$
	\begin{align}
		\upperq{\ell}{+}&\defeq \min\left\lbrace \sortQ{t} :~ \sortQ{t}> \sortQ{\ell},t\in[m]\right\rbrace\;,\label{eq:increase_bounds}\\
		\lowerq{\ell}{+}&\defeq \max\left\lbrace \sortQ{t} :~ \sortQ{t}\le \sortQ{\ell},t\in[m]\setminus \{\ell\}\right\rbrace\;, \nonumber\\
		\upperq{\ell}{-}&\defeq \min\left\lbrace \sortQ{t} :~ \sortQ{t}\ge \sortQ{\ell},t\in[m]\setminus \{\ell\}\right\rbrace\;,\nonumber\\
		\lowerq{\ell}{-}&\defeq \max\left\lbrace \sortQ{t} :~ \sortQ{t}< \sortQ{\ell},t\in[m]\right\rbrace\;.\label{eq:decrease_bounds}
	\end{align}
	If all entries of $\bQ$ are unique, then $\upperq{\ell}{+}=\upperq{\ell}{-}$ and $\lowerq{\ell}{+}=\lowerq{\ell}{-}$, so only if duplicate values exist, these bounds on $\sortQ{\ell}$ differ.
	An example clarifying these bounds is given in \cref{ex:qubo}.
	The idea is now that $\sortQ{\ell}$ is changed in such a way that it lies exactly in the middle between $\lowerq{\ell}{\pm}$ and $\upperq{\ell}{\pm}$. 
	For $\sortQ{1}<\sortQ{\ell}<\sortQ{m}$ we increase $\sortQ{\ell}$ if $\sortQ{\ell}-\lowerq{\ell}{-}<\upperq{\ell}{+}-\sortQ{\ell}$ and decrease otherwise. 
	The weight $\sortQ{\ell}$ is thus changed by
	\begin{align*}
		\wuh^{\MOR}=\begin{cases}
			\frac{\upperq{\ell}{+}-\lowerq{\ell}{+}}{2}-\min\{\upperq{\ell}{+}-\sortQ{\ell},\sortQ{\ell}-\lowerq{\ell}{+}\}, &\text{if } \sortQ{\ell}-\lowerq{\ell}{-}<\upperq{\ell}{+}-\sortQ{\ell}, \\
			\frac{\lowerq{\ell}{-}-\upperq{\ell}{-}}{2}+\min\{\upperq{\ell}{-}-\sortQ{\ell},\sortQ{\ell}-\lowerq{\ell}{-}\},  &\text{else}.
		\end{cases}
	\end{align*}
	For the edge case $\ell=m$, $\wuh^{\MOR}$ is given by
	\begin{align*}
		\wuh^{\MOR}=\begin{cases}
			\lowerq{m}{-}-\sortQ{m}+\minD{\mat Q},  &\text{if }\minD{\singleton{q_u: ~u\in [m]\backslash\singleton{m}}}=\minD{\mat Q}, \\
			\minD{\singleton{q_u: ~u\in [m]\backslash\singleton{m}}}-\minD{\mat Q}, &\text{else}.
		\end{cases}
	\end{align*}
	Similarly, for $\ell=1$
	\begin{align*}
		\wuh^{\MOR}=\begin{cases}
			\upperq{1}{+}-\sortQ{1}-\minD{\mat Q},  &\text{if }\minD{\singleton{q_u: ~u\in [m]\backslash\singleton{1}}}=\minD{\mat Q}, \\
			\minD{\mat Q}-\minD{\singleton{q_u: ~u\in [m]\backslash\singleton{1}}}, &\text{else}.
		\end{cases}
	\end{align*}
	
	Having the heuristic change $\wuh^{\cdot}$ at hand, we can determine the final change $\wut$ as
	\begin{align}
		\wut=\min\{\max\{\wuh^{\cdot},\wuo^-\},\wuo^+\}\;,\label{eq:final_change}
	\end{align} 
	which ensures that $\wut\in[\wuo^-,\wuo^+]$, such that the optimum is preserved.
	\begin{example}\label{ex:qubo}
		Consider the following exemplary \QUBO matrix with $n=3$:
		\begin{equation}
			\mat{Q}\defeq 
			\begin{bmatrix}
				-1 & 0.4 & 1 \\ 
				& 0.4 & -0.8 \\ 
				& & -1.5
			\end{bmatrix},
			\label{eq:example_matrix}
		\end{equation}
		with $\DR(\mat{Q})=\log_2 (2.5 / 0.2) = 3.64$.
		The ordering $\sortQ{1},\dots,\sortQ{9}$ is given by 
		\begin{equation*}
			\left(\sortQ{1},\dots,\sortQ{9}\right)=\left(-1.5,-1,-0.8, 0,0,0, 0.4, 0.4, 1\right)
		\end{equation*}
		and is visualized in \cref{fig:matrix_sorting}.
		\begin{figure}
			\centering
			\begin{subfigure}[t]{1\textwidth}
				\centering
				\begin{tikzpicture}
					\draw[latex-latex] (0, 0) -- (10, 0);
					\node at (10.2, 0.0) {$\mathbb{R}$};
					\path[fill=black] (0.455, 0) node[below=0.200] {$-1.5$} circle (3pt);
					\path[fill=black] (2.273, 0) node[below=0.200] {$-1$} circle (3pt);
					\path[fill=black] (3.000, 0) node[below=0.200] {$-0.8$} circle (3pt);
					\draw [decorate, decoration = {calligraphic brace, raise=2pt, amplitude=4pt}, thick] (2.273, 0.200) --  (3.000, 0.200) node[above=5pt,black, pos=0.5]{$\minD{\mat Q}$};
					\path[fill=black] (5.909, 0) node[below=0.200] {$0$} circle (3pt);
					\path[fill=black] (5.909, 0.300) circle (3pt);
					\path[fill=black] (5.909, 0.600) circle (3pt);
					\path[fill=black] (7.364, 0) node[below=0.200] {$0.4$} circle (3pt);
					\path[fill=black] (7.364, 0.300) circle (3pt);
					\path[fill=black] (9.545, 0) node[below=0.200] {$1$} circle (3pt);
					\draw [decorate, decoration = {calligraphic brace, raise=2pt, amplitude=4pt}, thick] (0.455, 0.900) --  (9.545, 0.900) node[above=5pt,black, pos=0.5]{$\maxD{\mat Q}$};
				\end{tikzpicture}
				\caption{We can read off $\maxD{\mat Q}=2.5$ and $\minD{\mat Q}=0.2$.}
				\label{fig:matrix_sorting:min_max}
			\end{subfigure}
			\begin{subfigure}[t]{1\textwidth}
				\centering
				\begin{tikzpicture}
					\draw[latex-latex] (0, 0) -- (10, 0);
					\node at (10.2, 0.0) {$\mathbb{R}$};
					\path[fill=black] (0.455, 0) node[below=0.200] {$\lowerq{1}{-}, \lowerq{1}{+}$} circle (3pt);
					\path[fill=black] (2.273, 0) node[below=0.360] {$\sortQ{1}$} circle (3pt);
					\path[fill=black] (3.000, 0) node[below=0.200] {$\upperq{1}{-}, \upperq{1}{+}$} circle (3pt);
					\path[fill=black] (5.909, 0) node[below=0.200] {$\lowerq{7}{-}$} circle (3pt);
					\path[fill=black] (5.909, 0.300) circle (3pt);
					\path[fill=black] (5.909, 0.600) circle (3pt);
					\path[fill=black] (7.364, 0) node[below=0.200] {$\upperq{7}{-},\sortQ{7},\lowerq{7}{+}$} circle (3pt);
					\path[fill=black] (7.364, 0.300) circle (3pt);
					\path[fill=black] (9.545, 0) node[below=0.200] {$\upperq{7}{+}$} circle (3pt);
				\end{tikzpicture}
				\caption{Bounds (\cref{eq:increase_bounds,eq:decrease_bounds}) on \QUBO parameters $\sortQ{1}=-1$ and $\sortQ{7}=0.4$.}
				\label{fig:matrix_sorting:bounds}
			\end{subfigure}
			\caption{Sorted \QUBO parameters of the matrix given in \cref{eq:example_matrix}. 
				Duplicates are indicated as vertically stacked points. }
			\label{fig:matrix_sorting}
		\end{figure}
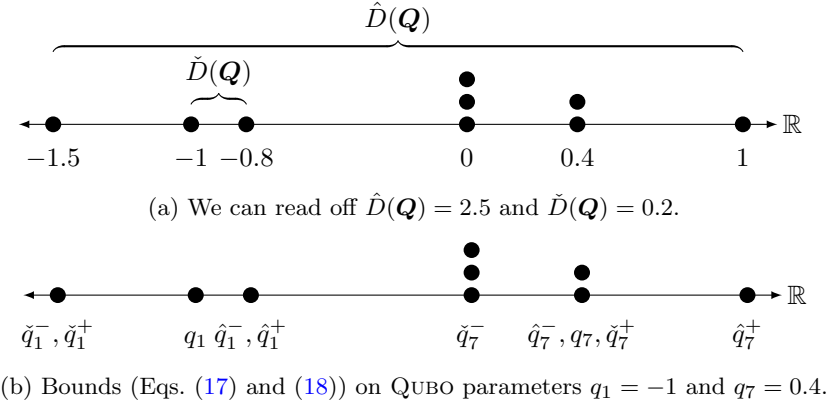
		The bounds from \cref{eq:increase_bounds,eq:decrease_bounds} for two specific parameters can be found in \cref{fig:matrix_sorting:bounds}.
		We want to increase the value $Q_{23}\equiv \sortQ{3}=-0.8$, because this would decrease $\minD{\mat Q}$ and thus decrease the dynamic range.
		We fix $k=2$, $l=3$ and find that $\vec{x}^*=(0, 1, 1)^{\top}$. 
		For maintaining the optimum $\vec{x}^*$ when changing $Q_{23}$, we need to obey \cref{theorem:optimum_bound}.
		Computing accurate bounds, \ie, the exact values, we obtain
		\begin{equation*}
			\check{y}_{00}=0,~ \check{y}_{01}=-1.5,~ \check{y}_{10}=0.4,~ \hat{y}_{11}=-1.9,
		\end{equation*}
		and thus $\wuo^{+}=\min\lbrace0,-1.5,0.4\rbrace- 1.9=0.4$.
		In words, we can maximally increase $Q_{23}$ by $0.4$ to maintain the optimum state, which is depicted in Fig.~\ref{fig:change:heuristics}.
		For decreasing the $\DR$ we have a look at the three heuristics $\MOR$, $\GRE$ and $\GRE_0$.
		The values $\wuh^{\MOR}=0.3$, $\wuh^{\GRE}=1.6$ and $\wuh^{\GRE_0}=0.8$ are also depicted in \cref{fig:change:heuristics}.
		\begin{figure}
			\centering
			\begin{subfigure}[t]{1\textwidth}
				\centering
				\begin{tikzpicture}
					\intervbc{3}{1.2}{1.455}{vir1of4}{$\wuo^{+}$}
					\intervbc{3}{0.9}{1.091}{vir2of4}{$\wuh^{\MOR}$}
					\intervbc{3}{0.6}{2.909}{vir3of4}{$\wuh^{\GRE_0}$}
					\intervbc{3}{0.3}{5.818}{vir4of4}{$\wuh^{\GRE}$ ($=\wud^+$)}
					
					\draw[latex-latex] (0, 0) -- (10, 0);
					\node at (10.2, 0.0) {$\mathbb{R}$};
					\path[fill=black] (0.455, 0) node[below=0.200] {$-1.5$} circle (3pt);
					\path[fill=black] (2.273, 0) node[below=0.200] {$-1$} circle (3pt);
					\path[fill=black] (3.000, 0) node[below=0.200] {$-0.8$}circle (3pt);
					\path[fill=black] (5.909, 0) node[below=0.200] {$0$}circle (3pt);
					\path[fill=black] (7.364, 0) node[below=0.200] {$0.4$}circle (3pt);
					\path[fill=black] (9.545, 0) node[below=0.200] {$1$} circle (3pt);
				\end{tikzpicture}
				\caption{Parameter value intervals: Maximum increase $\wuo^{+}=0.4$ (blue) of \QUBO parameter $Q_{23}$ such that the optimum is preserved, along with heuristic interval limits $\wuh^{\MOR}=0.3$ (cyan), $\wuh^{\GRE}=1.6$ (green) and $\wuh^{\GRE_0}=0.8$ (yellow).}
				\label{fig:change:heuristics}
			\end{subfigure}
			\begin{subfigure}[t]{1\textwidth}
				\centering
				\begin{tikzpicture}
					\draw[latex-latex] (0, 0) -- (10, 0);
					\node at (10.2, 0.0) {$\mathbb{R}$};
					\path[fill=vir2of4] (4.091, 0)  circle (3pt);
					\draw [decorate, decoration = {calligraphic brace, raise=2pt, amplitude=3pt, mirror}, thick] (3.9, -0.200) --  (4.282, -0.200) node[below=5pt,black, pos=0.5]{$\MOR$};
					\path[fill=vir4of4] (4.455, 0) circle (3pt);
					\draw [decorate, decoration = {calligraphic brace, raise=2pt, amplitude=3pt}, thick] (4.264, 0.200) --  (4.646, 0.200) node[above=5pt,black, pos=0.5]{$\GRE,\GRE_0$};
					\path[fill=black] (0.455, 0) node[below=0.200] {$-1.5$} circle (3pt);
					\path[fill=black] (2.273, 0) node[below=0.200] {$-1$} circle (3pt);
					\path[fill=black] (5.909, 0) node[below=0.200] {$0$} circle (3pt);
					\path[fill=black] (7.364, 0) node[below=0.200] {$0.4$} circle (3pt);
					\path[fill=black] (9.545, 0) node[below=0.200] {$1$} circle (3pt);
				\end{tikzpicture}
				\caption{New parameters after changing $Q_{23}=-0.8$ w.r.t. preserving the optimum: $\min\{\wuo^+,\wuh^{\MOR}\}=0.3$ for $\MOR$ (cyan), and $\min\{\wuo^+,\wuh^{\GRE}\}=0.4$ (yellow) for $\GRE$ (and $\GRE_0$ analogously). In both cases, the $\DR$ decreases by one bit, since $\minD{\mat Q}$ is doubled.}
				\label{fig:change:after}
			\end{subfigure}
			\caption{Change of \QUBO parameter $Q_{23}$.}
			\label{fig:change}
		\end{figure}
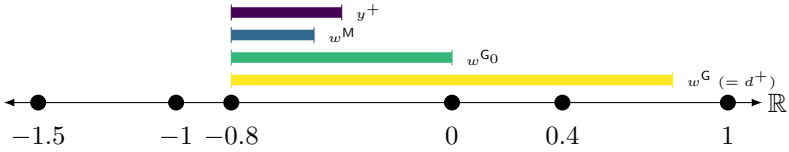
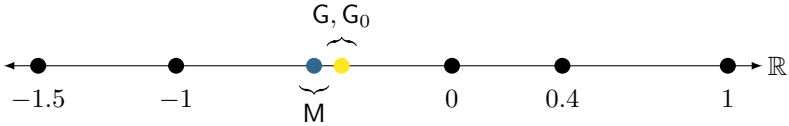
		We observe that $\MOR$ changes $Q_{23}$ to lie in the middle between its neighbors, $\GRE$ maximally increases $Q_{23}$ to maintain the $\DR$, while $\GRE_0$ sets $Q_{23}$ to $0$.
		In \cref{fig:change:after} the final changes are shown, using the three heuristics.
		Following \cref{eq:final_change}, the changes are given by $0.3$ for $\MOR$ and $0.4$ for $\GRE$ and $\GRE_0$, respectively.
		Both result in a doubling of $\minD{\mat Q}$ leading to a new dynamic range decreased by one bit, i.e., $2.64$.
	\end{example}
	
	\subsection{Choosing the Next Parameter}\label{sec:next_param}
	Since the changes of the \QUBO parameters are carried out in a successive fashion, it remains to decide which $k,l\in [n], k\le l$ to pick next.
	
	A very simple approach is to pick a random pair of indices, or iterate over all index pairs in sequence.
	Using this method, many -- or, with growing $n$, most -- iterations will not lead to a DR improvement, as only a few different parameters directly determine the DR, namely those closest and furthest apart (c.f. \cref{fig:matrix_sorting:min_max}).
	Conversely, changing such an ``inactive'' parameter can never lead to a decrease in DR, only to an increase.
	This realization leads to a better strategy, which is choosing only among those index pairs whose parameters determine DR.
	In our experiments, we compute their respective update values $a$ and greedily choose the one that leads to maximal DR reduction, breaking ties randomly.

	\section{Experiments}\label{sec:experiments}
	In this section we conduct experiments for showing the effectiveness of our proposed method.
	
	\subsection{Results for Random Instances}\label{sec:random_qubos}
	We conduct experiments with the parameter update methods presented in \cref{sec:reduce_methods}, namely
	\begin{itemize}
		\setlength{\itemindent}{2em}
		\item[$\GRE$] Greedily choose the parameter update that leads to the greatest DR decrease;
		\item[$\GRE_0$] Like $\GRE$, but prefer to set parameters to 0, if possible;
		\item[$\MOR$] Restrict bounds such that the parameter ordering remains intact.
	\end{itemize}
	\noindent
	The methods described in this article are implemented as part of our Python package \texttt{qubolite}\footnote{\url{https://github.com/smuecke/qubolite}}.
	
	Using $1000$ random \QUBO matrices with the entries being sampled uniformly in the interval $[-0.5, 0.5]$, we follow the different heuristics from \cref{sec:reduce_methods} for $1000$ iterations. 
	For preserving the optimum the upper bounds $\hat y_{ab}$ are computed using a local search and the lower bounds $\check y_{ab}$ are computed by using the roof-dual algorithm \cite{boros2008}.
	In addition we compare the two methods described in \cref{sec:next_param}, namely choosing the next \QUBO weight to be random or according to the largest impact on the decrease of the $\DR$.
	The $95\%$-confidence intervals are indicated.
	
	\cref{fig:dynamic_range_ratio} depicts the dynamic range ratio between the matrices arising by following the different heuristics and the original \QUBO matrix, that is $\DR(\mat Q+\mat{A})/\DR(\mat Q)$. 
	\begin{figure}[!t]
		\centering
		\includegraphics[width=1.0\textwidth]{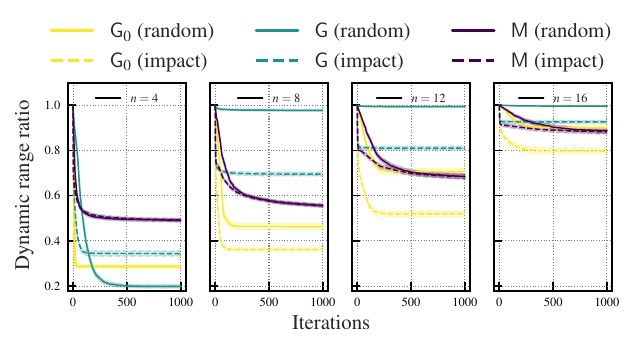}
		\caption{Dynamic range ratio for different \QUBO sizes $n\in\{4, 8, 12, 16\}$.}
		\label{fig:dynamic_range_ratio}
	\end{figure}
	We observe this ratio to be monotonically decreasing, indicating that our proposed heuristics are fulfilling their duty of reducing the $\DR$. 
	This decreasing percentage shrinks with larger $n$, \ie, it is harder to reduce the dynamic range for larger $n$.
	Generally, choosing the next \QUBO weight according to the myopic impact on the $\DR$ is better than a random decision (except for $n=4$), but also comes with a higher computational cost.
	With increasing $n$, it turns out that the $\GRE$ heuristic performs worse than the $\MOR$ heuristic. 
	However, the slight adaption $\GRE_0$ of setting \QUBO weights to $0$ if possible outperforms $\MOR$. 
	
	For further evaluation we first define a metric for comparing different rankings.
	\begin{definition}
		Let $\mat Q\in\QUBOSET_n$, and let $\vec x^1,\vec x^2,\dots,\vec x^{2^n}$ be the binary vectors of $\BB^n$ in lexicographic order.
		The \emph{induced ranking} of $\mat Q$ is a permutation $\pi_{\mat Q}\in [2^n]\rightarrow[2^n]$ such that \begin{equation*}
			f_{\mat Q}(\vec x^{\pi_{\mat Q}(i)})\leq f_{\mat Q}(\vec x^{\pi_{\mat Q}(i+1)}) ~\forall i\in[2^n-1]\;.
		\end{equation*}
	\end{definition}
	
	\begin{definition}
		Let $\pi,\pi':[K]\rightarrow[K]$ be two permutations for some $K>1$.
		The normalized \emph{Kendall $\tau$ distance} between $\pi$ and $\pi'$ is given by \begin{equation*}
			K_d(\pi,\pi')\defeq\frac{1}{2}+\frac{1}{K(K-1)}\sum_{\substack{i,j\in [K]\\ i<j}} \sign\left[(\pi(i)-\pi(j))\cdot(\pi'(i)-\pi'(j))\right]\;.
		\end{equation*}
	\end{definition}
	Intuitively, this distance measures the proportion of disagreement between rankings over all pairs of indices, \ie, the percentage of $(i,j)$ such that $\pi(i)<\pi(j)$, but $\pi'(i)>\pi'(j)$, and vice versa.
	If $K_d(\pi,\pi')=0$, then $\pi$ and $\pi'$ are identical, and if $K_d(\pi,\pi')=1$, then $\pi'$ is the reverse of $\pi$.
	Therefore, Kendall $\tau$ distance gives us a measure of how much the DR reduction ``scrambles'' the value landscape of $f_{\mat Q}$ by computing $K_d(\pi_{\mat Q},\pi_{\mat Q+\mat A})$, which is relevant, \eg, for the performance of local search heuristics.
	
	\cref{fig:state_ordering} shows the Kendall $\tau$ distance between the induced rankings of the \QUBO instances on $\BB^n$ before and after their $\DR$ reduction.
	\begin{figure}[!t]
		\centering
		\includegraphics[width=1.0\textwidth]{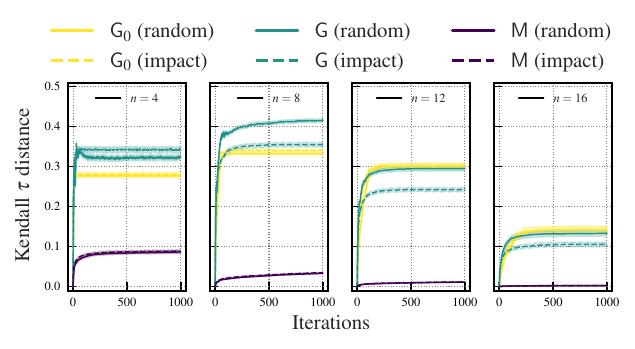}
		\caption{State ordering for different \QUBO sizes $n\in\{4, 8, 12, 16\}$.}
		\label{fig:state_ordering}
	\end{figure}
	We observe that $\GRE$ and $\GRE_0$ ``scramble'' the ordering of binary vectors more strongly than $\MOR$.
	This is exactly what we expected, as $\MOR$'s objective is to preserve the ordering of parameters, which in turn leads to more conservative parameter modifications.
	This reduces the overall ``noise'' that is added to the energy landscape.
	
	Instead of considering the \QUBO values themselves, we depict the change of the \QUBO weight ordering in \cref{fig:qubo_weight_ordering}. 
	Again, the Kendall $\tau$ distance is used and we observe similar effects to \cref{fig:state_ordering}. 
	We can see that $\MOR$ maintains the weight ordering.
	\begin{figure}[!t]
		\centering
		\includegraphics[width=1.0\textwidth]{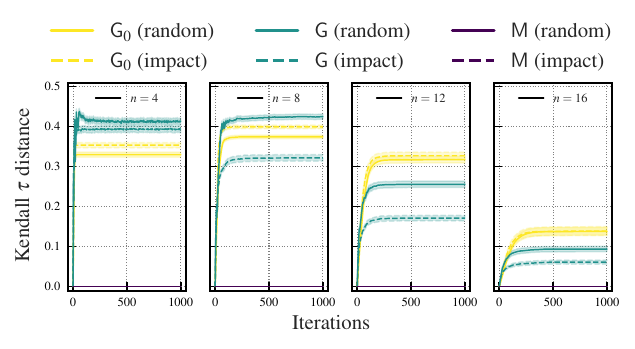}
		\caption{Distance between \QUBO parameter orderings for different sizes $n\in\{4, 8, 12, 16\}$.
			Note that the line for method $\MOR$ (random) is constantly 0.}
		\label{fig:qubo_weight_ordering}
	\end{figure}
	
	Finally, \cref{fig:unique_params_ratio} displays the unique weight ratio, that is the number of unique \QUBO weights of the current iteration divided by the number of unique \QUBO weights of the original \QUBO. $\MOR$ never changes the uniqueness of the weights, while the ratio is monotonically reduced for $\GRE^0$, since many weights are set to $0$. With using $\GRE$, the number of unique weights is first reduced but then starts to increase when the to be changed weights are chosen randomly. 
	\begin{figure}[!t]
		\centering
		\includegraphics[width=1.0\textwidth]{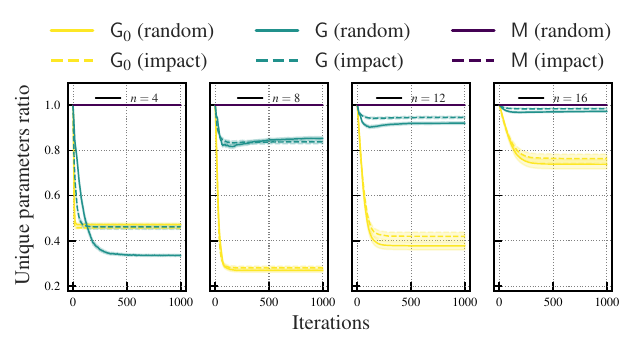}
		\caption{Percentage of unique \QUBO parameter values for different sizes $n\in\{4, 8, 12, 16\}$.
			Note that the line for method $\MOR$ (random) is constantly 1.}
		\label{fig:unique_params_ratio}
	\end{figure}

	\subsection{Results on Quantum Hardware}
	Lastly, we want to investigate to what extend our DR reduction method can help to improve the performance of actual quantum hardware.
	To this end, we follow these steps: \begin{enumerate}
		\item Generate a \QUBO instance $\mat{Q}$
		\item Apply our DR reduction method to obtain $\mat{Q}'$
		\item Apply quantum annealing to both $\mat{Q}$ and $\mat{Q}'$, performing multiple readouts
		\item Compare the occurrence probabilities of the global minimum
	\end{enumerate}
	
	\begin{figure}
		\centering
		\includegraphics[width=\textwidth]{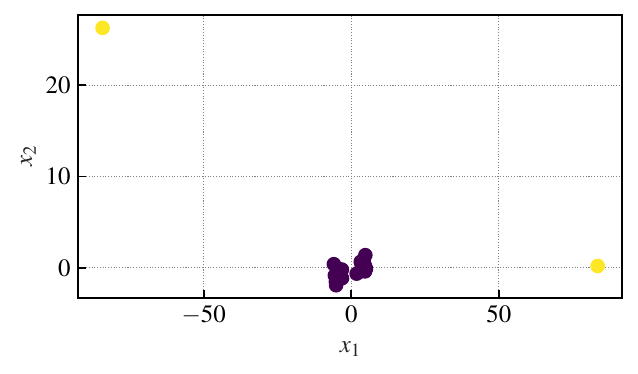}
		\caption{2D data set used for \BC, consisting of 20 points, each with two coordinates.
			The two outliers (shown in a different color) lead to large values in the Gram matrix, from which the \QUBO parameters are derived, leading in turn to a large dynamic range.}
		\label{fig:dwave-data}
	\end{figure}
	
	As two exemplary problems, we perform \BC and \SuS.
	\BC stands for ``binary clustering'' and is an unsupervised ML task, where data points are assigned to one of two classes (``clusters'').
	\SuS consists of finding a subset from a list of values that sum up to a given target value.
	Both have well-established \QUBO embeddings \cite{bauckhage2018,biesner2022}.
	
	To generate data for \BC, we sample i.i.d. $n=20$ 2-dimensional points from an isotropic standard normal distribution.
	Then we create two clusters by applying $(x_1,x_2)\mapsto (x_1-4,x_2)$ to the first ten points, and $(x_1,x_2)\mapsto (x_1+4,x_2)$ to the last ten.
	As a last step, we choose the points 1 and 19 and multiply their coordinates by 20, which leads to a data set containing two outliers, shown in \cref{fig:dwave-data}.
	
	\begin{figure}
		\centering
		\begin{subfigure}{.5\textwidth}
			\centering
			\includegraphics[height=4.5cm]{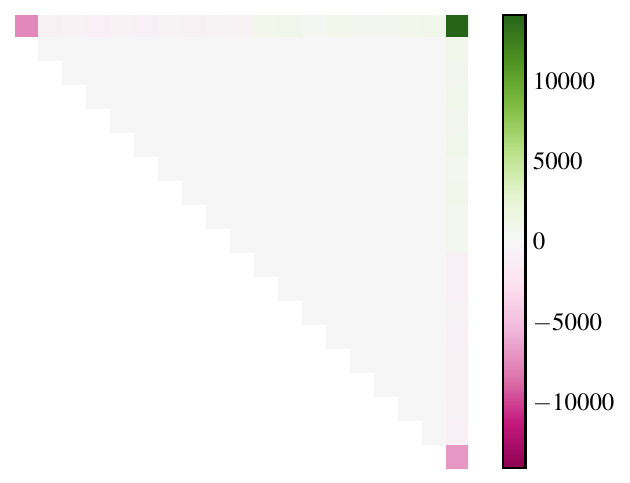}
			\subcaption{\BC\\(DR=22.7921)}
			\label{fig:dwave-qubo}
		\end{subfigure}%
		\begin{subfigure}{.5\textwidth}
			\centering
			\includegraphics[height=4.5cm]{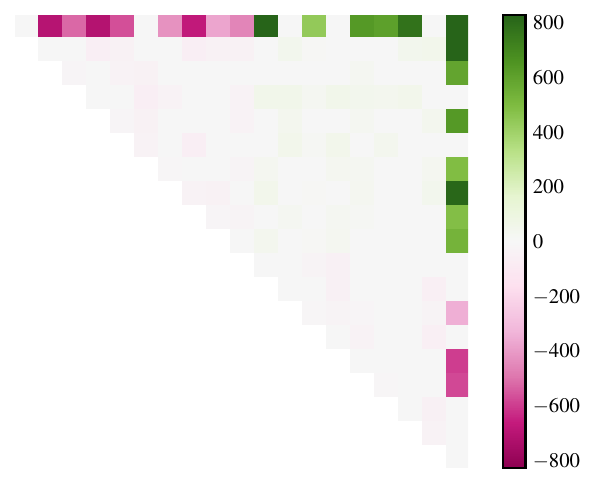}
			\subcaption{\BC\\(reduced, DR=11.0427)}
			\label{fig:dwave-qubo-compr}
		\end{subfigure}
		\begin{subfigure}{.5\textwidth}
			\centering
			\includegraphics[height=4.5cm]{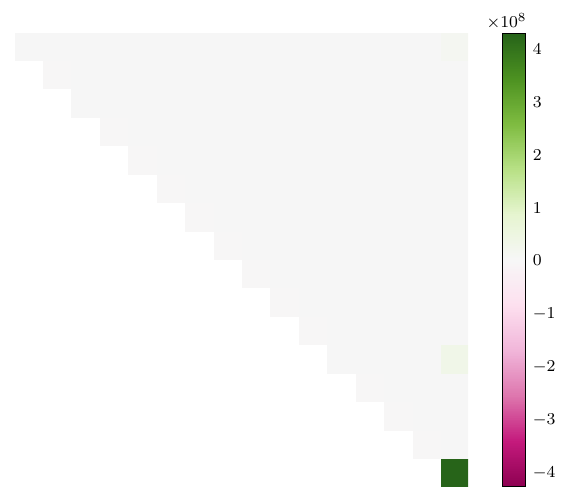}
			\subcaption{\SuS\\(DR=25.6765)}
			\label{fig:dwave-ss-qubo}
		\end{subfigure}%
		\begin{subfigure}{.5\textwidth}
			\centering
			\includegraphics[height=4.5cm]{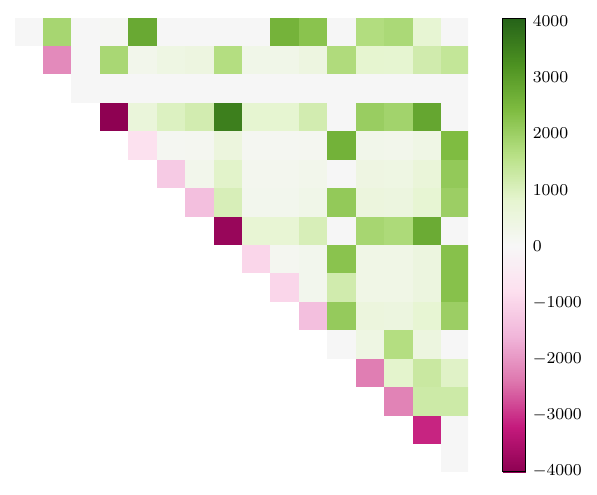}
			\subcaption{\SuS\\(reduced, DR=9.8882)}
			\label{fig:dwave-ss-qubo-compr}
		\end{subfigure}
		\caption{\QUBO parameter matrices for \BC and \SuS problems, before (left) and after (right) applying our proposed dynamic range reduction method.
			Notice that the difference in DR is illustrated by the color scale, which renders most parameter values in the original \QUBO matrices indistinguishable, which is greatly improved on the right.}
	\end{figure}
	
	From this data we derive a \QUBO instance $\mat{Q}$ using the method from \cite{bauckhage2018}.
	We use a linear kernel, which leads to a vanilla 2-means clustering based on Euclidean distance.
	To ensure this problem has only one optimal solution, we assign class 0 to point 20 and only optimize over the remaining 19 points (otherwise there would be two symmetrical solutions).
	The resulting parameter values are shown in \cref{fig:dwave-qubo}.
	As is apparent from the color scale, the DR is very high (22.7921 bits), which makes most values near 0 completely indistinguishable.
	
	Now, we apply our $\GRE_0$ DR reduction method to $\mat{Q}$, since it was the most promising one apparent from \cref{sec:random_qubos}.
	Again we use a local search and the roof-dual algorithm for computing optimum-preserving bounds.
	To limit execution time, we give our algorithm a fixed budget of 100 iterations.
	The resulting \QUBO instance $\mat{Q}'$ is shown in \cref{fig:dwave-qubo-compr}; clearly, much more detail is visible, as the color scale is much narrower, already hinting at the lower DR.
	Indeed, we find that $\DR(\mat{Q}')=11.0427$, which is a reduction by more than half.
	
	Similarly, we sample a \SuS problem, using the same approach as for \cref{fig:rounding-error} described in \cref{sec:rounding-error-method}.
	We use the same algorithm configuration as before, but allow for 150 iterations.
	The \QUBO parameter matrices are shown in \cref{fig:dwave-ss-qubo,fig:dwave-ss-qubo-compr}.
	The DR is $25.6765$ before and $9.8882$ after.
	
	Next, we attempt to solve all four \QUBO instances on a D-Wave quantum annealer.
	To this end, we use the Advantage system 4.1, which we access through the Python interface.
	The D-Wave annealers have a fixed connectivity structure, i.e., only a subset of qubit pairs can be assigned a weight.
	Therefore, dense \QUBO problems (like ours) must be embedded into this connectivity graph structure through redundant encoding and additional constraints.
	To improve comparability, we have this embedding computed once for the original \QUBO and re-use it for the compressed one.
	This is possible because they have the same size, differing only in parameter magnitudes.
	
	\begin{figure}
		\centering
		\includegraphics[width=\textwidth]{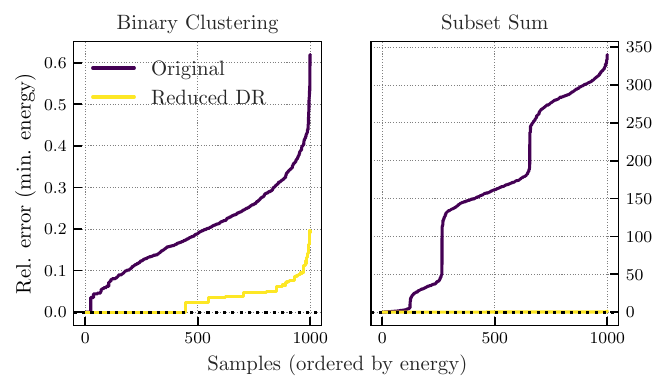}
		\caption{Relative energies of 1000 samples obtained from the D-Wave quantum annealer for the original (dark blue) and reduced (yellow) \QUBO instances.
			The samples are sorted in ascending order of energy.
			The relative error is computed w.r.t. the ground-truth minimal energy computed by brute force; lower is better.
			We observe that reducing DR increases the prevalence of the global optimum and yields more low-energy solutions in general.}
		\label{fig:dwave-energies}
	\end{figure}
	
	We perform $1000$ \emph{reads} for each \QUBO instance and record their energy values.
	Recall that our proposed DR reduction method only keeps the solution vectors intact, but changes the parameter values, therefore the energy values are no longer comparable after reduction.
	For this reason, we compute the ground-truth minimal energies $v^*$ for both $\mat{Q}$ and $\mat{Q}'$, and plot the \emph{relative deviation} of the measured energies $v$ w.r.t. these values, \begin{equation*}
		\left|\frac{v^*-v}{v^*}\right|\;,
	\end{equation*}%
	along the y-axis in \cref{fig:dwave-energies}.
	
	As is apparent from the figure, the DR-reduced \QUBO instances yield the global optimum much more frequently than the original instances of high DR.
	We verified that the vector of minimal energy in the compressed \QUBO instances indeed yield the minimal energy of the original problems, confirming that our method preserves the optimum correctly.
	To quantify the improvement, we consider the energy values of samples obtained from the compressed \QUBO instances w.r.t. the original (uncompressed) \QUBO. 
	By inspecting the total number of samples with an energy equal to or lower than the lowest sampled energies obtained from the original \QUBO instances, we find that the prevalence of low-energy solutions is $17.76$ ($237$) higher for the compressed \BC (\SuS) \QUBO instances.
	We conclude that our method helps a quantum annealer to find the optimum more reliably, since the ICEs limit the dynamic range of the hardware parameters \cite{dwaveice2023}.
	
	\section{Conclusion}
	While \QUBO is arguably one of the most prominent and well-studied problem classes, many of their properties are not yet well understood. Our results fall into the cutset of 
	theory and practice by providing a formalization of practical limitations w.r.t. rounding-induced perturbations of the problem parameters. Based on our formalization, we propose a completely novel methodology for reducing the precision required to encode \QUBO instances. 
	More precisely, 
	we employ the notion of dynamic range to characterize \QUBO instances w.r.t. the number of bits needed to faithfully encode their parameters. 
	We show that a high dynamic range indicates that the parameters simultaneously cover a large value range while also containing fine gradations. 
	Further, we defined the optimum inclusion relation $\optleq$, which formalizes the idea that the set of minimizing solutions of a \QUBO instance $\mat{Q}$ is a subset of those minimizing $\mat{Q}'$.
	This implies that a minimizing solution of $\mat{Q}$ is guaranteed to also minimize $\mat{Q}'$.
	Based on scale-invariance properties, we derive a theoretical bound on the minimal scaling factor required such that, after rounding, the optimal solution remains unchanged. 
	
	Applying our findings in a naive manner results in an intractable method due to the unavailability of the true minimal values. 
	Nevertheless, 
	we establish upper and lower bounds, that allow us to compute intervals in which parameter values can be modified without affecting the minimizing vectors. Hence, leaving optimal solutions intact. 
	This opens a space of \QUBO instances over which we can try to minimize the dynamic range while preserving the minima.
	To this end, a greedy algorithm is devised that iteratively modifies parameter values while reducing the problem's dynamic range. 
	On top of the bare algorithm, we investigate different selection strategies for choosing the next entry of the \QUBO parameter matrix.
	
	In a series of experiments we put our theoretical framework into practice.
	First, we randomly sample \QUBO instances and apply our algorithm with different parameter selection strategies.
	We observe that all of them lead to a decrease of DR, the most effective one being the greedy strategy which tries to set parameters to zero, if possible ($\GRE_0$).
	In two more experiments, we took one binary clustering problem and one \textsc{SubsetSum} problem, reduced the resulting \QUBO matrices' DR, and solved both the original and the compressed problem on a D-Wave quantum annealer.
	We find that our compression method greatly improves the performance of QA, leading to a significantly higher probability to observe the global optimum.
	
	We uncovered numerous intriguing open questions while working on this article, which may serve to inspire further research.
	
	The algebraic notion of optimum inclusion can be naturally extended to an equivalence relation, which induces a partitioning on $\QUBOSET_n$.
	The number of equivalence classes must necessarily be finite and bounded above by $2^{2^n}-1$, which implies there is only a finite number of meaningfully different \QUBO instances in terms of their set of minimizing bit vectors.
	A fascinating research endeavor is to find, for each equivalence class, the representative of lowest DR, or---even better---an algorithm to do so.
	
	In the context of QA, the spectral gap, which we used for our theoretical bound, is crucial to the duration of the annealing process.
	It has been proven that for certain random problem instances, annealing time becomes exponentially long, rendering QA just as inefficient as brute force \cite{altshuler2009}.
	As we have shown in this article that DR has an impact on the result quality as well, it would be interesting to investigate the interrelation between the two quantities.
	
	The algorithm we proposed for reducing DR contains heuristic elements worth improving upon.
	We are currently investigating %more elaborate 
	AI-based parameter selection strategies, \eg, putting the DR reduction task in a Reinforcement Learning framework.
	In summary, we provide an efficient pre-processing technique that can be applied out-of-the-box to arbitrary \QUBO instances to improve their feasibility. 
	
	\section*{Statements and Declarations}
	\subsection*{Author Contributions}
	All authors contributed to the conception and design.
	The experiments were conducted by Sascha M\"ucke and Thore Gerlach.
	The first draft was written by Sascha M\"ucke and Thore Gerlach.
	All authors commented on and revised previous versions of the manuscript, and approved the final version.
	
	\subsection*{Competing Interests}
	The authors of this article declare that there are no competing interests.
	
	\subsection*{Funding}
	This research has been funded by the Federal Ministry of Education and Research of Germany and the state of North-Rhine Westphalia as part of the Lamarr-Institute for Machine Learning and Artificial Intelligence.
	
	\FloatBarrier

	\newpage
	\appendix
	\section{Methodology for \cref{fig:rounding-error}}
	\label{sec:rounding-error-method}
	
	In this section we describe the experimental setup used to obtain the results shown in \cref{fig:rounding-error}.
	
	We used random instances of the \textsc{SubsetSum} problem.
	For this problem, we are given a list $A=(a_1,\dots,a_n)$ of $n$ integers and a target value $T$.
	Our task is to find a subset $S\subseteq [n]$ such that $\sum_{i\in S}a_i=T$.
	This problem lends itself naturally to \QUBO, where we use $n$ binary variables which indicate if $i\in S$ for each $i$.
	The energy function is then simply \begin{equation*}
		f(\vec{x})=\biggl(\sum_ia_ix_i-T\biggr)^2 \propto \sum_{i,j} a_ia_jx_ix_j -2T\sum_i a_ix_i\;,
	\end{equation*}%
	which yields the \QUBO parameters \begin{equation*}
		Q_{ij}=\begin{cases}2a_ia_j &\text{if } i\neq j\\
			a_i^2-2Ta_i &\text{otherwise.}\end{cases}
	\end{equation*}
	
	For our experiment, we set $n=16$ and sampled the elements of $A$ i.i.d. as $\abs{\nint{10\cdot Z}}$, where $Z$ follows a standard Cauchy distribution.
	This leads to occasional outliers with large magnitudes, which in turn yields \QUBO instances with large DR.
	
	Next, we sample the number of summands $k$ from $\nint{U}$ where $U$ follows a triangular distribution with parameters $a=\frac{n}{5}$, $b=\frac{n}{2}$ and $c=\frac{4n}{5}$, so that, on average, half of the elements of $A$ contribute to the sum.
	Finally, we sample $k$ indices from $[n]$ without replacement to obtain $S$, and set $T=\sum_{i\in S}a_i$.
	This way, we handily generate problems where we already know the global optimum.
	
	Following this process we generate $N=100,000$ \QUBO instances $\mathcal{D}=\lbrace\mat{Q}^i\rbrace_{i\in[N]}$ and compute their DR.
	From the empirical distribution of DR values, we compute the $\frac{i}{5}$-quantiles, which we label $\beta_i$, for $i\in\lbrace 0,\dots,5\rbrace$.
	For each $i\in[4]$, we define bins $B_i=[\beta_{i-1},\beta_i)$, and $B_5=[\beta_4,\beta_5]$.
	Now, we partition our data set by sorting the \QUBO instances into the 5 bins by their DR, $\mathcal{D}_i\defeq\lbrace\mat Q^j:~j\in[N],\,\DR(\mat{Q}^j)\in B_i\rbrace$.
	For each $\mathcal D_i$, we scale and round each instance to the number of bits indicated by the x-axis.
	
	The y-axis shows ``optimum correctness'', which indicates the proportion of rounded \QUBO instances that retained a global optimum after scaling and rounding.
	To compute this value, firstly let $correct(Q'\given Q)\defeq 1$ if $\mat Q'\optleq\mat Q$ and $0$ otherwise.
	We check this by obtaining $v^*=\min_{\vec x}~\vec x^\top\mat Q\vec x$ and $\tilde{\vec x}^*\in\optset{\mat Q'}$ by brute force; if $f_{\mat Q}(\tilde{\vec x}^*)=v^*$, then necessarily $\mat Q'\optleq\mat Q$.
	For each number of bits $n_b$ and bin index $i$, the y-value we plot is $\frac{1}{20,000}\sum_{\mat Q\in\mathcal D_i}correct(\tilde{\mat Q}^{(n_b)}\given \mat Q)$, where $\tilde{\mat Q}^{(n_b)}$ denotes the version of $\mat Q$ scaled and rounded to a bit precision of $n_b$.
	
	Intuitively, this tells us the approximate probability that, for a randomly sampled \textsc{SubsetSum} problem, we still obtain the correct solution when parameter precision is reduced to the given number of bits.
	Clearly, this probability increases when we use more bits to encode the parameters.
	This is expected, as higher precision leads to more accurate representation of the energy landscape.
	Further, the result shows that the probability decreases with higher DR, which we also expected:
	A high DR indicates that there are fine distinctions between parameter values, making it more likely for rounding to collapse parameters and, consequently, lose energy distinctions between solution vectors.
\end{document}